%% file: TGMDEcon2.tex
\documentclass[12pt, oneside,reqno]{amsart}
\usepackage[left=2.5cm,right=2.5cm, top=2.5cm, bottom=2.5cm, marginparwidth=1.75cm]{geometry}
\usepackage{caption ,subcaption, tikz, color, amsfonts, mathtools, amsmath,mathrsfs, amsthm, todonotes, stmaryrd}
\usepackage[backref=page]{hyperref}
\usepackage{cite, verbatim} 
\usetikzlibrary{patterns}

\usepackage[
alphabetic,
msc-links,
nobysame,
lite,
]{amsrefs} 

\title{Tropical Geometry and Mechanism Design}
\author{Robert Alexander Crowell} 
\address[Robert Alexander Crowell]{Hausdorff Center for Mathematics, Bonn 53115, Germany}
\email{crowellr@uni-bonn.de}
\author{Ngoc Mai Tran}
\address[Ngoc Mai Tran]{Department of Mathematics, University of Texas at Austin, Texas TX 78712, USA and Hausdorff Center for Mathematics, Bonn 53115, Germany}
\email{ntran@math.utexas.edu}

\DeclareMathOperator{\R}{\mathbb{R}}

\DeclareMathOperator{\N}{\mathbb{N}}
\DeclareMathOperator{\covec}{\mathsf{\underline{coVec}}_T} 
\DeclareMathOperator{\Plift}{\mathcal{P}^{\uparrow}}

\theoremstyle{plain}
\newtheorem{theorem}{Theorem}[section]
\newtheorem*{theorem*}{Theorem}
\newtheorem{lemma}[theorem]{Lemma}
\newtheorem{proposition}[theorem]{Proposition}
\newtheorem{corollary}[theorem]{Corollary}
\theoremstyle{definition}
\newtheorem{definition}[theorem]{Definition}
\newtheorem{example}[theorem]{Example}

\newtheorem{remark}[theorem]{Remark}

\newcommand{\1}{\mathbf{1}}
\newcommand{\ti}{t_{-i}}

\newcommand{\cf}{\emph{cf.~}}

\newcommand{\Eig}{\mathsf{\underline{Eig}}}
\newcommand{\TA}{\mathbb{TA}}

\newcommand{\tconv}{\overline{\mathsf{tconv}}}

\newcommand{\minH}{\underline{\mathcal{H}}}
\newcommand{\maxH}{\overline{\mathcal{H}}}
\newcommand{\minL}{\underline{\mathcal{L}}}
\newcommand{\maxL}{\overline{\mathcal{L}}}
\newcommand{\cov}{\underline{\mathsf{coVec}}}

\newcommand{\Ti}{T_{-i}}

\newcommand{\img}{{[g]}}
\newcommand{\cardg}{{\llbracket  g\rrbracket}}

\date{\today}

\begin{document}

\begin{abstract}
We develop a novel framework to construct and analyze finite valued, multidimensional mechanisms using tropical convex geometry. We geometrically  characterize incentive compatibility using cells in the tropical convex hull of the type set. These cells are the sets of incentive compatible payments and form tropical simplices, spanned by generating payments whose number equals the  dimension of the simplex. The analysis of the collection of incentive compatible mechanisms via tropical simplices and their generating payments facilitates the use of geometric techniques. We use this view to derive a new geometric characterization of revenue equivalence but also show how to handle multidimensional mechanisms in the absence of revenue equivalence. 
\vskip .1in

\noindent {\bf Keywords:} Multidimensional Mechanism Design, Tropical Geometry, Tropical Convexity, Tropical Combinatorics, Incentive Compatibility, Cyclical Monotonicity, Revenue Equivalence, Extremal Payments.\\

\noindent
\emph{2010 Mathematics Subject Classification.}  91A80, 14T05. 
\emph{JEL Classification. } D82, C02. 
\end{abstract}
\maketitle

\section{Introduction}

Mechanisms are engineered games, devised to implement outcomes that depend on the private information of individuals in the economy, called their \emph{type}. We consider deterministic mechanisms with $m\in \N$ outcomes and quasi-linear utilities. In this setting types can be represented by a set $T\subset \R^m$. A \emph{mechanism} is specified by pair of an outcome function $g:T \to\{1, \ldots m\}$ and a payment vector $p\in\R^m$, taking a possibly false report $s\in T$ of an agent whose true type is $t^*$ to an outcome and a transfer, yielding utility $\displaystyle t^*_{g(s)}-p_{g(s)}$. An outcome function $g$ is called \emph{incentive compatible} (IC) if there is a payment vector $p$ such that truth-telling is an utility maximizing strategy. 

Multidimensional mechanism design remains a hard topic, mainly because constructive techniques known from single-dimensional settings no longer apply in more general models, e.g. \cite{milgrom2002envelope,mcafee1988multidimensional}. In this paper we develop geometric techniques to handle deterministic multidimensional mechanisms using tropical convex geometry. 

Rochet's classical characterization \cite{rochet1987necessary} states that an outcome function $g$ is IC if and only if it is \emph{cyclically monotone}, meaning that the induced \emph{type graph} on node set $T$ with edge weights $l_{s,t}^g=t_{g(s)}-t_{g(t)}$ contains no negative weight cycles \cite{rockafellar1966characterization, levin1999abstract, gui2004dominant, vohra2011mechanism}. 
This characterization is implicit since it analyzes mechanisms via network flows in digraphs induced by the outcome function. That is, given a type space $T$ and an outcome function $g$, deciding whether $g$ is IC is possible by inspecting the type graph. However, given a type space $T$ and digraph all of whose cycles have non-negative weights, we lack an applicable technique to decide whether this graph arose from some outcome function on $T$, and if so, which outcome functions induced it. The difficulty is that cycle conditions do not encode agent's choices explicitly, and thus do not facilitate defining IC outcome functions from them, even though efforts have been undertaken to  interpret these conditions economically, e.g. \cite{rahman2010detecting}.

Another obstacle is the calculation of IC payments. Even if a given outcome function $g$ is known to be IC, the characterization of its IC payments $\mathcal P(g)$ is complex. It is well-known that $\mathcal P(g)$ is a polytope described by $m(m-1)$ hyperplanes, one for each ordered pair of outcomes. However, using this to gain a structural understanding of a mechanism's payments proves difficult. Efforts to find usable characterizations of IC payments are ongoing, \cite{carbajal2013mechanism, kos2013extremal}. For instance, Kos and Messner \cite{kos2013extremal} gave `extremal payments' which bound all IC payments from above and below, but which are not sufficient to characterize $\mathcal P(g)$.  
Absent a general characterization, the case in which  $\mathcal{P}(g)$ consist of a single payment up to addition of constants, known as \emph{revenue equivalence} (RE), remains an essential ingredient for mechanism design, as it has been from the outset \cite{myerson1981optimal, klemperer1999auction}.
A characterization of revenue equivalence for outcome functions relying on the type graph was given by Heydenreich, M\"uller, Uetz and Vohra \cite{heydenreich2009characterization}. In fact, the extremal payments of \cite{kos2013extremal} can be seen as a refinement of their result. Criteria on type spaces to force RE for any IC outcome function have been obtained by Chung and Olszewski \cite{chung2007non-differentiable}, but are quite involved.  However, RE precludes any re-distributive motives of information rents once the outcome function has been fixed, an assumption that is restrictive in some economic contexts, see \cite{carbajal2013mechanism,kos2013extremal} for a discussion and references.

In this paper we develop a novel geometric view on the study of mechanisms rooted in tropical geometry, thereby offering new perspectives on these problems. Tropical convex geometry was developed to analyze combinatorial aspects of network flow problems geometrically. Our core characterization result, Theorem \ref{thm:main.generic}, shows that there is a one-to-one correspondence between the fundamental economic concepts of IC outcome functions and their payments, and the basic objects of tropical convex geometry: tropical polytopes, covectors and tropical eigenspaces. This simple observation provides a new constructive characterization of all IC outcome functions together with their IC payments which is not based on cyclical monotonicity, but instead on the \emph{tropical polytope} generated by $T$, a polyhedral complex defined from the type space whose cells naturally encode economic incentives. This novel view is foundational for all our results. To discuss them, let us draw the parallels to the problems we identified above. 

Our approach provides the means to use cyclical monotonicity conditions constructively. The \emph{allocation matrix} $L^g$ of the outcome function $g:T\to\{1, \ldots, m\}$ is obtained from the weighted digraph used in Rochet's theorem by quotienting zero weight cycles, \cite{vohra2011mechanism}. Given a matrix $L$ and a type space $T$, we show in Theorem \ref{thm:realizable} how to decide whether there exist an outcome function $g$ such that $L = L^g$, and if so, how to explicitly construct these outcome functions based on the geometry of the type space $T$. 

Another contribution is the complete characterization IC payments. Consider again the polytope $\mathcal P(g)$ of IC payments of $g$. An alternative to its hyperplane characterization would be to present its vertices, from which all other payments could be written as convex combinations. Unfortunately, these do not have a simple characterization and their number can range from $1$ to $\binom{2m-2}{m-1}$, \cite[Prop. 19]{develin2004tropical}. In contrast, in tropical algebra, $\mathcal P(g)$ is always a \emph{tropical simplex}. It can be written as the tropical convex hull of a unique minimal set of generating elements whose number equals the dimension of $\mathcal P(g)$. We term these elements the \emph{tropical generating payments} and geometrically characterize them and their number in Theorems \ref{thm:tropicaleigenspacegenerator} and \ref{thm:dimension}. 
This result demonstrates that tropical algebra is well-suited for working with IC payments, by intertwining algebraic and geometric techniques, providing a perspective not available prior. It allows us to explicitly incorporate the geometry of $T$ into the analysis in ways that algebraic  manipulation of incentive constraints does not.  
As a consequence we obtain Corollary \ref{cor:re}, which provides a novel characterization of RE outcome functions. It refines the result on RE type spaces of \cite{chung2007non-differentiable} and gives an explicit, geometric characterization of RE outcome functions, whose implicit and algebraic counterpart was established in \cite{heydenreich2009characterization}. We elaborate on these connections in Section \ref{sec:re.details}. 

While rooted in the mathematics of tropical geometry, the objects we define naturally represent economic incentives. Here we also argue that tropical geometry provides a complementary picture to abstract cycle conditions for mechanism design. It is well known that network flows can be cast in terms of tropical linear algebra \cite{baccelli1992synchronization, butkovivc2010max}, both are very powerful tools to analyze a given outcome function \cite{vohra2011mechanism}. Indeed, in Section \ref{sec:payments.algebra}, we utilize classical results on tropical eigenspaces to derive new results on generating IC payment. However, tropical geometry has tools from combinatorics, integer programming, and algebraic geometry, which do not have an immediate network flow interpretation. Section \ref{sec:informal} contains an informal discussion with examples to demonstrate the complementary geometric, algebraic and combinatorial views on mechanisms tropical mathematics provides, and how we utilize them to derive our results.  For instance, Theorem \ref{thm:main.generic}, was achieved through the combinatorics of tropical hyperplane arrangements. As a corollary, we show that the number of IC outcome functions on any generic type space with $m$ outcomes and $r$ types is exactly $\binom{r+m-1}{m-1}$. That is, while the exact set of outcome functions depends on the type space, their number does not. Such results would be difficult to obtain from the network flow approach alone.

\subsection*{Organization} Section \ref{sec:informal} begins with an informal example to demonstrate our main points. This section is fully self-contained. In Section \ref{sec:overview} we recall the basic terminology of mechanisms and tropical convex geometry. This section also introduces our running example, Example~\ref{ex:covectors}. In Section \ref{sec:main}, we collect the main results of our geometric theory, and provide an extensive discussion of its economic interpretation. In Section \ref{sec:realizability} we give a new result on the characterization of allocation matrices. For clarity and ease of exposition, technical generalizations to multi-player settings, generic and non-compact type spaces are given in Section \ref{sec:extensions}. All proofs are collected in Section \ref{sec:proofs}. We conclude with a summary in Section~\ref{sec:conclusion}.

\subsection*{Further Literature}While this paper is self-contained, it is not intended as an introduction to tropical geometry. The mathematics we build on was pioneered by \cite{develin2004tropical} and further developed in \cite{fink2013stiefel,joswig2016weighted}, but tropical convex geometry and tropical linear algebra have a much longer history, see \cite{butkovivc2010max,baccelli1992synchronization}.
For a fuller picture, see the monographs \cite{butkovivc2010max, baccelli1992synchronization} on solving tropical linear equations, or \cite{joswig20xxessentials} on tropical combinatorics, and \cite{maclagan2015introduction} on tropical algebraic and convex geometry. We refer to \cite{ziegler2012lectures} for the basic terminology of classical polyhedral geometry. The works of \cite{baldwin2012tropical, baldwin2015understanding, tran2015product-mix} are the first to apply tropical geometry to microeconomic theory. Other applications include \cite{shiozawa2015exotic, joswig2016cayley} to trade theory and   \cite{akian2012tropical} to  mean-payoff games. These works attest to fruitful interactions between tropical geometry and economics. 

\subsection*{Notation} For an integer $n$, define $[n] := \{1, 2, \ldots, n\}$. If $g:T\to [m]$ is an outcome function, then we shall denote its image in $[m]$ by $\img$, and the cardinality of this set by $\cardg$.
We use the underline notation, such as $\underline{\oplus}, \underline{\odot}, \underline{\mathcal{H}},\ldots$ to indicate objects defined with arithmetic done in the min-plus algebra, and the overline notation $\overline{\oplus}, \overline{\odot}, \overline{\mathcal{H}},\ldots$ to indicate the same objects defined with arithmetic in the max-plus tropical algebra. Identify a graph with its incidence matrix.

\section{An Informal Overview and Example}\label{sec:informal} We begin with an informal discussion. This section is fully self-contained and can be skipped if the reader prefers delving into the formal definitions right away. Tropical mathematics provides two complementary perspectives on mechanism design, one via tropical linear algebra, the other via tropical convex geometry. It is the interplay of these views that distinguishes the tropical approach from other techniques. 

Fix a type space $T\subset \R^m$. If $g:T\to [m]$ is an outcome function we denote by $\mathcal P(g)\subset \R^m$ the set of IC prices, i.e. vectors $p\in \R^m$ that satisfy $t_{g(t)} - p_{g(t)} \geq t_{g(s)}-p_{g(s)}$ for all $t,s \in T$. Conversely, given a price vector $p\in \R^m$ we denote by $\mathcal G({p})$ the set of outcome functions $g:T\to [m]$ for which $p$ is an IC price. 
First, we wish to understand the maps $p\mapsto \mathcal G(p)$ and $g\mapsto \mathcal P(g)$ using tropical hyperplane arrangements. To simplify the discussion, suppose that the outcome function $g: T \to [m]$ is onto and that $T$ is compact. Let $V_t: \R^m \to \R$ be the function that maps a payment vector $p \in \R^m$ to the negative of the equilibrium  utility\footnote{We chose the negative of the utility to be consistent with the sign convention in the tropical and cyclical monotonicity literature.}, that is,
$$ V_t(p) = - \max_{s\in T}(t_{g(s)}-p_{g(s)})=-\max_{i \in [m]}(t_i-p_i) = \min_{i\in[m]}(p_i - t_i).$$ 
This is a piecewise linear concave function on $\R^m$, obtained as the point-wise minimum of $m$ classical hyperplanes. In tropical min-plus algebra $(\R,\underline \oplus,\odot)$ where tropical addition of numbers is their  minimum, $a\underline \oplus b= \min\{a, b\}$, and tropical multiplication is the usual addition, $a\odot b=a+b$ for $a, b\in \R$, the function $V_t$ becomes formally linear
\begin{equation}\label{eqn:pz.min}
V_t(p) = (-t)^\top \underline{\odot} \, p =\underline \bigoplus_{i=1}^m (-t_i)\odot p_i.
\end{equation}
The set of points where $V_t$ is not smooth is called the min-plus tropical hyperplane with apex~$t$, denoted $\minH(-t)$. In economics terms these are prices $p\in \R^n$ at which the agent of type $t$ would be indifferent between two or more outcomes. Conversely, at any price $p$ not in the tropical hyperplane an agent of type $t$ has a unique preferred outcome. The union of all min-plus hyperplanes with apices in $T$ is called the min-plus arrangement of $T$, and is denoted $\minH(-T):=\bigcup_{t\in T}\minH(-t)$. This set partitions $\R^m$ into polyhedral regions called cells.
All prices $p$ in the relative interior of a given cell $\sigma$ induce the same set of preferred outcomes for all types $t \in T$. Therefore, all prices $p\in \sigma$ must have the same set of IC outcome functions $\mathcal{G}(p)$. Conversely, each IC outcome function $g$ has a payment $p \in \R^m$, which belongs to some cell. Thus, one can infer the maps $p \mapsto \mathcal{G}(p)$ and $g \mapsto \mathcal{P}(g)$ from $\minH(-T)$. This is the main content of Proposition \ref{prop:covector} and the key bridge between tropical geometry and incentive compatibility. 

\input{TGMDfigure_firstexample.tex}

\begin{example}[Cells and outcome functions]\label{ex:motionation.1} Consider $m=3$ outcomes. Notice that tropical hyperplanes and thus the cells of the tropical min-plus arrangement are closed under addition of constants.  This leads to the definition of \emph{tropical affine space} $\TA^2$, whose elements are classes consisting of vectors in $\R^3$ which are tropical scalar multiples of each other, i.e. differ by addition of constants. In economics terms, types depicted in $\TA^2$ are \emph{incentive types} rather than payoff types. Incentive types distill the economics of incentive compatibility. Explicitly a point $(x_1, x_2, x_3)\in \R^3$ projected to the coset $[(x_1, x_2, x_3)]\in \TA^2$ can be depicted as $(x_2-x_1, x_3-x_1)$ in $\R^2$ by normalizing the first coordinate to zero. Thus we may depict figures in two instead of three dimensions.
Fix a type set $T=\{t^1, \ldots, t^6 \}\subset \R^3$ represented by the small black dots in Figure \ref{fig:motivation} (\textsc{a}). For instance, the payoff type $t^1=(1, 1.5, 2)\in \R^3$ is represented in the figure as the class $[(1, 1.5, 2)]=[(0, 0.5, 1)]$ in $\R^2$. At each point we also depicted a min-plus tropical hyperplane using fine black lines. Each of the full-dimensional regions in the figure is a cell. By the rational above the full-dimensional cells in this example are in bijection with the IC outcome functions on $T$. Moreover, each of these cells equals the IC payment set of the associated outcome function. It turns out that the full-dimensional, bounded cells correspond to onto mechanisms, while the outcome functions corresponding to full-dimensional, unbounded cells are not onto. Now consider the gray shaded cell and the payment vector $p=[(p_1, p_2, p_3)]=[(0,4,2)]$ defined up to an additive constant. The set of its IC outcome functions $\mathcal G (p)$ contains a unique IC $g:T\to [3]$ which is given by $$g(t^1) = g(t^2)=1,\quad g(t^3)=g(t^4)=2, \quad g(t^5)=g(t^6)=3.$$ The construction of outcome functions from cells follows from combinatorial data associated to the arrangement known as \emph{covectors}, defined in Section \ref{sec:connections}, but can also be carried out geometrically as in Example~\ref{ex:construction1}, below.
\end{example}

The cells of the min-plus arrangement can be considered from a more abstract and algebraic viewpoint. Each bounded cell of the arrangement is a \emph{tropical simplex} which enjoys algebraic structure. Tropical simplices admit a unique minimal set of generators whose min-plus linear span exactly equals the simplex. While each tropical simplex when viewed as an ordinary convex polytope is also spanned as the classical convex hull of its vertices, these vertices are difficult to characterize, and their number can range anywhere between $1$ and $\binom{2m-2}{m-1}$, \cite[Prop. 19]{develin2004tropical}. In contrast, in tropical algebra the generators are easily geometrically characterized and their number always equals the dimension of the tropical simplex in $\R^m$, \cf Theorem \ref{thm:dimension}. These complementary views between tropical convex geometry and tropical linear algebra are not available for the orthodox techniques applied in the study of mechanisms. In particular the equality between the number of algebraic generators and the geometric dimension is a unique feature of tropical arithmetic. 

The significance of these observations for economics become evident upon noting that all IC payment sets are tropical simplices upon choosing definitions carefully, namely restricting to the \emph{tropical polytope of $T$}, i.e. the bounded cells of $\minH(-T)$. This allows us to give a unique minimal set of \emph{generating payments} for any outcome function sufficient to characterize all its IC payments, \cf Theorem \ref{thm:tropicaleigenspacegenerator}. The task of characterizing IC payments was picked up by \cite{kos2013extremal, carbajal2013mechanism}. Kos and Messner \cite{kos2013extremal} defined extremal payments which bound any IC payment from above and below, but generally prove to be not a sufficient characterization.

\begin{example}[Generating payments]
	Consider again the gray cell in Figure \ref{fig:motivation} which is precisely $\mathcal P(g)$ for $g$ from Example \ref{ex:motionation.1}. Noting that this cell has dimension two in $\TA^2$ and thus dimension three in $\R^3$, it has three generating payments which can be represented by the gray  points labeled $v_1=(0,5,1), v_2=(0,2,2), v_3=(0,6,4)$ in Panel (\textsc{b}). These are sufficient to fully characterize all IC payments for $g$  as follows. Algebraically, for any $q  \in \mathcal P(g)$ there exist constants $c_1, c_2, c_3\in \R$ such that $q=c_1\odot v_1 \underline \oplus c_2\odot v_2 \underline \oplus c_3\odot v_3$, moreover any such linear combination is an IC payment for $g$. For instance, we may write $$p=(0,4,2)=(1\odot v_1 )\underline \oplus (2\odot v_2) \underline \oplus v_3=(1, 6, 2) \underline \oplus (2, 4, 4) \underline \oplus (0, 6, 4).$$ Geometrically, the gray cell equals the union of bounded cells of the \emph{max-plus} tropical hyperplane arrangement on $v_1, v_2, v_3$, depicted with black lines in Panel (\textsc{b}), \cf Remark \ref{rem:duality}.  This example also shows that generally two payments are not sufficient to characterize all IC payments. Indeed, in Example \ref{ex:information.rents.2} we give a geometric construction of the extremal payments of Kos and Messner, which can be represented by the points labeled $v_3$ and $r$ in Figure~\ref{fig:motivation}~(\textsc b). Upon normalizing appropriately, the payment vector labeled $s$ is bounded point-wise above and below by $v_3$ and $r$, respectively, but is not IC for $g$. Moreover, $\mathcal P(g)$ equals the ordinary convex hull of 6 points, which matches the upper bound $\binom{4}{2}$, however, the reader will check that the IC payment sets of other onto outcome functions on $T$, i.e. full-dimensional bounded cells may have fewer vertices. This is not the case for tropical arithmetic, where all full-dimensional bounded cells have exactly 3 tropical generators. \end{example}
 
\input{TGMDfigure_re.tex}

A very special case considered in the literature is that of \emph{revenue equivalence} (RE). This occurs exactly when there is only one generating payment. The link between the dimension of cells and their number of generators equips us with a simple but powerful strategy to give a geometric characterization of RE by analyzing the dimension of cells in the min-plus arrangement on $T$. Our characterization of RE is presented in Corollary \ref{cor:re}.

\begin{example}[RE via dimension of cells]
To demonstrate our geometric view of RE consider Figure \ref{fig:re}. There the type space $T$  consists of three parallel lines. Panel (\textsc{a}) depicts the cells whose associated mechanisms are RE, panel (\textsc{b}) depicts the cells whose associated mechanisms are not RE. The union of the gray regions coincides with the tropical polytope generated by $T$, which can be identified with the IC payments of any IC outcome function on $T$. Using approximations we see that payment simplices highlighted in panel (\textsc{a}) converge to points, while those highlighted in panel (\textsc{b}) converge to vertical lines, as schematically indicated. For instance, the payments represented by the two black dots in panel (\textsc{b}) can be used as IC payments of the same outcome function, yet differ not just by a constant. 
\end{example}
Using our view of RE we give geometric meaning to a characterization of RE of \cite{heydenreich2009characterization}. There it is shown that RE of an outcome function is equivalent to the anti-symmetry of the graph-distance in its type graph. This connection is formalized in Corollary \ref{cor:re.function} and serves once more to manifest the tight connection between tropical algebra and tropical geometry. Importantly however, our result is formulated in terms of a geometric property of cells of the tropical polytope which are \emph{a priory} not related to outcome functions, while the characterization in \cite{heydenreich2009characterization} is formulated as a property of type graphs induces by IC outcome functions. It is therefore not clear how assumptions on the geometry of type spaces affect the RE property. A recent result emphasizing a geometric view is due to \cite[Theorem 4]{chung2007non-differentiable}, which states that a type space $T$ is RE if and only if there do not exist disjoint subsets $B_1, B_2$, a function $r: B_1 \cup B_2\to \R$, and an $\epsilon >0$, such that $T$ equals the union of two non-empty sets $V_+(B_1,\epsilon,r)$ and $V_-(B_2,\epsilon,r)$, whose definitions depend on the parameters given. Yet their result pertains to type spaces as a whole and provides only an incomplete understanding of type spaces on which RE and non-RE outcome functions coexist. Our result can be applied using a straight-forward approximation technique and checking geometrically that certain cells of the tropical polytope converge to points, providing a simple algorithm that is applicable to any type space. It applies to subsets of all IC outcome functions and by extension to type spaces globally providing us with a clear geometric understanding of the RE property. 

To conclude the discussion let us briefly return to the problem of deciding whether a given type graph to which Rochet's theorem is to be applied, arose on a fixed  type space $T$ and if so, which outcome functions induced it. Let us call such graphs \emph{realizable} on $T$. To answer the realizability question we rely on the dual representation of tropical simplices by hyperplanes. Indeed, each simplex possesses a dual representation via tropical hyperplanes. Somewhat more surprisingly the same simplex can be represented in min-plus as well as max-plus tropical arithmetic. Our characterization of realizability in Theorem \ref{thm:realizable} crucially relies on these dual geometric  representations of the simplices in $\minH(-T)$ enabling us to explicitly relate to the geometry of the type space in question. 

These examples illustrate that tropical mathematics is uniquely suited to  analyze the cells of $\minH(-T)$ which encode preferences and incentives to gain economic insights that would be difficult to obtain using only established techniques. In particular tropical techniques provide the geometric tools that allow us to visualize mechanisms and easily construct explicit examples.

\section{Mechanism Design and Tropical Geometry: an overview} \label{sec:overview}
\subsection{Mechanisms and incentive compatibility}\label{sec:background}
Consider a game with one agent and $m~\in~\mathbb{N}$ possible outcomes. Fix $T \subset \R^m$ of possible preferences, called the \emph{type space}, where $t_i$ measures the value of outcome $i$. To avoid technicalities in pathological cases, we shall assume that $T$ is compact. The theory for general $T$ is discussed in Section \ref{sec:extensions}.

A \emph{mechanism} is a pair~$(g, p)$ consisting of an outcome function $g: T \to [m]$ and a payment vector $p \in \R^m$. At the beginning of the game, Nature chooses a true type $t^*\in T$ and communicates it privately to the agent.  The agent's action is then to declare a type~$s \in T$ to the mechanism, which may be different from the true type $t^\ast$. Upon announcing  $s$, the game ends with outcome $g(s)$, and payment $p_{g(s)}$ by the agent to the mechanism.
The agent, knowing $(g,p)$, will declare a type $s \in T$ that maximizes utility, 
\[u([g(s),p(s)], t^\ast) =  t^\ast_{ g(s)}
 - p_{g(s)}.\]
Mechanisms under which the utility maximizing strategy is truth-telling are called \emph{incentive compatible}.

\begin{definition}\label{defn:IC}
A mechanism $(g,p)$ is \emph{incentive compatible} (IC) if regardless of $t^\ast$ the agent always maximizes utility by declaring the true type. That is,
\begin{equation}\tag{IC}\label{eqn:single.type.simple}
	t^\ast_{g(t^*)}
 - p_{g(t^*)}\geq t^\ast_{ g(s)}
 - p_{g(s)} \quad  \mbox{ for all } \quad s,t^\ast \in T.
\end{equation}
An outcome function $g: T \to [m]$ is incentive compatible if there exists $p \in \R^m$ such that $(g,p)$ is IC. 
\end{definition}

We denote by $\img \subset [m]$ the image of $g$. The set of IC payments is defined to be 
\begin{equation}\label{eqn:pg}
\mathcal{P}(g) = \{(p_i: i \in \img) \in \R^{\img} \mbox{ such that } (g,p) \mbox{ is IC for some } p \in \R^m\}.
\end{equation}
For a fixed $p\in \R^m$, the set of \emph{IC outcome functions} supported by $p$ is
\begin{equation}\label{eqn:gp}
\mathcal{G}(p) = \{g: T \to [m] \mbox{ such that } (g,p) \mbox{ is } IC\}.
\end{equation}
For a fixed type space $T$, the set of \emph{IC outcome functions on $T$} is
$$ \mathcal{G}= \{g: T \to [m] \mbox{ such that } g \mbox{ is } IC\}. $$

\subsection{Connections to tropical convex geometry}\label{sec:connections}

We now formalize the discussion from Section \ref{sec:informal}. In the \emph{min-plus algebra} $(\mathbb{R}\cup\{+\infty\}, \underline{\oplus},\odot)$ addition and multiplication are defined by
$$a \underline{\oplus} b := \min(a,b), \quad a \odot b := a + b\quad  \mbox{ for } a,b \in \R\cup\{+\infty\}.$$
Analogously, the \emph{max-plus algebra} $(\mathbb{R}\cup\{-\infty\}, \overline{\oplus},\odot)$ is defined by
$$a \overline{\oplus} b := \max(a,b), \quad a \odot b := a + b\quad  \mbox{ for } a,b \in \R\cup\{-\infty\}.$$
All definitions stated in the min-plus algebra above have analogous definitions in the max-plus algebra. As we shall see, for certain economic concepts it is more natural to work with max, while for others it is more natural to work with min. Therefore we retain both algebras in this paper. 

As discussed above, a fixed $t \in \R^m$ defines the min-plus linear function $V_t$ in (\ref{eqn:pz.min}), interpreted as the negative of the equilibrium utility for an agent with true type $t$. For $I \subseteq [m]$, \emph{sector $I$ of the min-plus arrangement at $t$} is defined to be the closed cone
$$ \underline{\mathcal{H}}_I(-t) := \{z \in \R^m: z_i - t_i \leq z_j - t_j \mbox{ for all } i \in I, j \in [m]\}. $$
Clearly for any pair $I,J \subseteq [m]$, we have $\minH_{I\cup J}(-t) = \minH_I(-t) \cap \minH_J(-t)$, so the collection of these closed cones forms a polyhedral fan in $\R^m$, \cite[Sec. 7.1]{ziegler2012lectures}, called the \emph{min-plus arrangement of $t$}, and denoted $\minH(-t)$. For a subset $T \subset \R^m$ the \emph{min-plus arrangement of $T$} is the polyhedral complex obtained as the refinement of the fans $\minH(-t)$ for all $t \in T$,
$$ \minH(-T) := \bigwedge_{t\in T}\minH(-t). $$
Each closed cell $\sigma$ of $\minH(-T)$ is obtained as the intersection  \begin{equation}\label{eqn:sector.intersection} \sigma = \bigcap_{t \in T} \minH_{I(\sigma,t)}(-t)
\end{equation}  
where $I(\sigma,t) \subseteq [m]$ for each $t \in T$. The tuple $(I(\sigma,t): t \in T)$ uniquely specifies $\sigma$, and is called its \emph{min-plus covector}, denoted $\cov(\sigma)$. The \emph{min-plus covector of a point} $p \in \R^m$, denoted $\cov(p)$, is the covector of the smallest cell that contains $p$. The following lemma gives a convenient way to think of covectors as a collection of bipartite graphs.
\begin{lemma}\label{lem:cov.p}
For $p \in \R^m$, $\cov(p)$ is the bipartite graph with nodes $[m] \times T$, where $(i,t)\in [m]\times T$ is an edge of this graph if and only if $p\in \minH_i(-t)$. 
\end{lemma}
Identify an outcome function $g: T\to [m]$ with its graph $g$ on $[m] \times T$, where $g(t) = i$ if and only if $(i,t)$ is an edge of the graph $g$. The following proposition links the covector of $p$ with respect to $\minH(-T)$ to $\mathcal{G}(p)$. As discussed in Section \ref{sec:informal}, this is the fundamental geometric connection between incentive compatibility problems and tropical geometry.
 
\begin{proposition}\label{prop:covector}
Let $g: T \to [m]$ be an outcome function, and $p \in \R^m$ be a price vector. Then $g \in \mathcal G(p)$ if and only if $g$ is a subgraph of $\mathsf{\underline{coVec}}_T(p)$.
\end{proposition}

We now turn to the study of IC payments via the \emph{max-plus} tropical convex hull of the type set $T$. The appearance of max instead of min in this name is the consequence of the trivial min-max duality, namely, $\min(a,b) = -\max(-a,-b)$. 
\begin{definition}\label{def:tropical.hull}
The \emph{max-plus convex hull} generated by $T$, denoted $\tconv(T)$, is polyhedral complex consisting of all cells in $\minH(-T)$ which are bounded up to addition of constants. 
\end{definition}
This name is motivated by the fact that for a finite set $T$, $\tconv(T)$ coincides with the max-plus polytope generated by $T$, defined as the  max-plus linear span of the points in $T$ \cite[Prop. 4]{develin2004tropical}.
While every payment induces some IC outcome function, the following key result states that only payments that lie in $\tconv(T)$ are interesting for economic applications. 

\begin{definition}\label{def:p.lift}
Let $g: T \to [m]$ be an outcome function. The \emph{cell of $g$} is the unique maximal cell of $\tconv(T)$ whose covector contains $g$ as its subgraph. That is,
\begin{equation}\label{eqn:pg.tropical}
\Plift(g) := \{q \in \tconv(T): q \in \mathcal{\underline{H}}_i(-t) \mbox{ for all } (i, t) \in g\}.
\end{equation}
If no such cell exists, define $\Plift(g) =\emptyset$.
\end{definition}
Unlike $\mathcal{P}(g)$, which is a polytope in $\R^M$ for some $M < m$ when $g$ is not onto, the cell $\Plift(g)$ is always a polytope in $\R^m$ when $T$ is compact. 

Proposition \ref{prop:projection.payments} below shows that there is an isomorphism between $\mathcal{P}(g)$ and $\Plift(g)$. As we shall prove below, this isomorphism has many desirable properties. Most importantly, it implies that the polyhedral complex $\tconv(T)$ represents \emph{all possible} IC payment sets of \emph{all possible} IC outcome functions on $T$.

\begin{proposition}\label{prop:projection.payments}
Let $g: T \to [m]$ be an outcome function and $\img \subseteq [m]$ its image. The projection $\pi_g: \R^m \to \R^\img$ defined by
$$ \pi_M(p) = (p_i: i \in \img) $$
is an affine isomorphism from $\Plift(g)$ to $\mathcal{P}(g)$. In particular, the dimensions of $\mathcal P(g)$ and $\Plift(g)$ agree.
\end{proposition}
\begin{corollary}
The outcome function $g$ is IC if and only if its cell $\Plift(g)$ is non-empty.
\end{corollary}

Suppose we pick a vector $p \in \R^m$ and utilize Proposition \ref{prop:covector} to infer the set $\mathcal{G}(p)$ of all IC outcome functions supported by $p$. In general, $p$ needn't be in the set $\bigcup_{g \in \mathcal{G}(p)}\Plift(g)$ of all IC payments compatible with these outcome functions. By Proposition \ref{prop:projection.payments} this cannot happen if $p\in \tconv(T)$. However, even then, this set may not be a cell in the tropical polytope. This is only guaranteed if $\mathcal{G}(p)$ is a singleton. The following proposition shows when this is the case. It is an essential stepping stone for our effort to establish a well-defined correspondence between and outcome functions. The proposition also highlights the roles of the covector in finding $g$, and the role of $\tconv(T)$ in finding $\Plift(g)$.

\begin{proposition}\label{prop:covec.full.dim}
Let $p \in \R^m$. Then $\mathcal{G}(p)$ contains a single IC outcome $g$ if and only if $p$ lies in the interior of a full-dimensional cell $\sigma$ of $\minH(-T)$. In this case, $g$ is given by
\begin{equation}\label{eqn:g.sigma}
 g(t) = i \iff (i, t) \in \covec(\sigma),
\end{equation}
and $\Plift(g) = \sigma \cap \tconv(T)$.
\end{proposition}

\begin{remark}\label{rem:duality}
The switch between min and max, from hyperplanes to the convex hull is a consequence of $\min(a,b)=-\max(-a, -b)$. It is not to be confused with the duality theorem between the $V$ and $H$-representation of tropical polytopes, which states that tropical max-plus polytopes can be written as the intersection of min-plus hyperplanes, as well as the max-plus convex hull of a unique minimal generating set, called its tropical vertices \cite{develin2004tropical,gaubert2011minimal,joswig2005tropical}.   Similarly min-plus polytopes have a $V$- and $H$-representation. 
In fact, the set of IC payments $\mathcal P(g)$ is both a min- and a max-plus polytope \cite{joswig2010tropical, joswig2016weighted}. Its $H$-representation as a \emph{max-plus} polytope is precisely the cell of $g$. In contrast, its $V$-representation as a \emph{min-plus} polytope is given in Theorem \ref{thm:tropicaleigenspacegenerator}.
\end{remark}

\begin{remark}\label{rem:affine.space}
All tropical arrangements have a common lineality space $$\R\cdot(1,\ldots, 1) = \R \odot (0, \ldots, 0).$$ Hence so does the tropical polytope and the set of IC payments. In tropical geometry, it is customary to work modulo this lineality space. This leads to the $(m-1)$-dimensional \emph{tropical affine space} $\TA^{m-1}$,
$$\mathbb{TA}^{m-1} \equiv \R^m / \R \cdot (1, \ldots, 1).$$ 
For figures, we shall identify $\TA^{m-1}$ with $\R^{m-1}$ via the homeomorphism 
\begin{equation}\label{eqn:tp.to.r}
[x]:=\{a \odot (x_1, \ldots, x_m): a \in \R\} \in \TA^{m-1} \mapsto (x_2-x_1, \ldots, x_m-x_1) \in \R^{m-1}.
\end{equation}
Tropical affine space is the natural ambient space to study incentive problems. In economic terms, this means that only relative valuation of the agent matter for truthfulness. The notion of `type' in $\TA^{m-1}$, which is formally a coset, is more accurately described as \emph{incentive type}. However, we will not afford this distinction trusting that no confusion will arise. 
\end{remark}

\subsection{Examples}
Examples \ref{ex:arrangement} and \ref{ex:covectors} below illustrate the definitions of sectors, tropical arrangements and covectors. Example \ref{ex:construction1} foreshadows our main results. 

\input{TGMDfigure_hyperplanes.tex}

\input{TGMDfigure_arrangement.tex}

\input{TGMDfigure_construction.tex}

\section{Main results}\label{sec:main}

\subsection{Global structure of IC outcome functions  and payments}
Theorem \ref{thm:main.generic} below gives a characterization for $\mathcal{G}$ and $\mathcal{P}$ in terms of $\minH(-T)$ and $\tconv(T)$. This result solves for the global structure of the collection of all IC outcome functions on $T$. To avoid technicalities, we state the case when $T$ is tropically generic. The general case can be reduced to this case through generic perturbations, and is treated in Section \ref{sec:extensions}. A set $T \subset \R^{m}$ is \emph{tropically generic} if there is no subset of $2 \leq k \leq m$ points in $T$ whose projection onto $k$ coordinates lie on a tropical hyperplane in $\R^{k}$. For example, finitely many points picked independently from some continuous distribution in $\R^m$ are almost surely generic.

\begin{theorem}\label{thm:main.generic}
Fix a tropically generic type space $T \subset \R^m$. Then $g$ is an IC outcome on $T$ if and only if it equals $\covec(\sigma)$ for some full-dimensional cell $\sigma$ of $\minH(-T)$. In this case, $\Plift(g) = \sigma \cap \tconv(T)$. 
\end{theorem}
\begin{corollary}\label{cor:cell.mechanism.bijection}
Fix a tropically generic type space $T \subset \R^m$. Then $g\in \mathcal G$ is an onto IC outcome function if and only if its corresponding cell $\sigma$ is a full-dimensional cell of $\tconv(T)$. 
\end{corollary}

Proposition \ref{prop:covec.full.dim} shows that each full-dimensional cell of $\minH(-T)$ gives rise to a unique IC outcome function $g$. The above theorem completes the characterization by showing that on generic $T$ any IC outcome function is associated to a cell. Thus we obtain a bijection between full dimensional cells and IC outcome functions. 
Furthermore, when $T$ is generic, the number of full-dimensional cells of $\minH(-T)$ and $\tconv(T)$ only depends on the cardinality of $T$ and the dimension $m$ \cite[Cor. 25]{develin2004tropical}. Consequently, if two generic type spaces have an equal number of types, then so is the number of IC, and onto IC outcome functions.

\begin{corollary}\label{cor:cardinality}
Suppose $T \in \R^m$ is tropically generic and contains $r$ points. Then the number of IC outcome functions on $T$ equals $\binom{r+m-1}{m-1}$, and the number of onto IC outcome functions on $T$ equals $\binom{r+m-4}{m-1}$.
\end{corollary}

\input{TGMDfigure_onto.tex}

\subsection{Tropical algebra of IC payments}\label{sec:payments.algebra}
We now apply classical theorems on the decomposition of tropical eigenspaces to obtain new results on the structure of $\mathcal{P}(g)$. Let $g:T\to [m]$ be an outcome function with image set $\img \subseteq [m]$, whose cardinality we denote by $\cardg$. The \emph{allocation matrix} $L^g \in \R^{\img \times \img}$ of $g$ is defined by
\begin{equation}\label{eqn:allocation.matrix}
L^g_{jk} := \inf_{t \in g^{-1}(j)}\{t_j - t_k\} \quad \mbox{ for all } j,k \in \img.
\end{equation} 
Algebraic manipulations show that 
\begin{equation}\label{eqn:pg.allocation}
\mathcal{P}(g) = \{p \in \R^\img: p_i - p_j \leq L^g_{ij}, \forall~i,j\in \img\} = \{p \in \R^\img:  \min_{j=1,\ldots,m} L^g_{ij} + p_j = p_i, \forall \, i \in \img\}.
\end{equation}
Rewriting (\ref{eqn:single.type.simple}) with matrix-vector and scalar-vector multiplication carried out in the min-plus algebra we obtain \begin{equation}\label{eqn:single.min.plus}
\Eig(L^g)=\{p \in \R^\img: L^g \,\,\underline{\odot}\,\, p = 0\,\, \underline{\odot}\,\, p\},
\end{equation}
which is the min-plus eigenspace of the allocation matrix $L^g$ with eigenvalue zero. Therefore, one obtains the following key result.

\begin{corollary}\label{cor:payments}
Let $g: T \to [m]$ be an outcome function. Then $\mathcal P(g)$ equals the min-plus eigenspace of $L^g$ with eigenvalue zero.  
\end{corollary}

It turns out that the min-plus eigenvalue of a $L^g$ matrix is unique and equals the minimum normalized cycle in the graph with edge weights $L^g$ \cite{Cg62}. Therefore, one recovers Rochet's cyclical monotonicity condition for incentive compatibility \cite{rochet1987necessary}, namely, that $g$ is IC if and only if all cycles of $L^g$ are nonnegative.

\begin{theorem}[\cite{Cg62}, \cite{rochet1987necessary}]\label{thm:rochet}
An outcome function $g$ is IC if and only if the allocation matrix $L^g$ has min-plus tropical eigenvalue zero, or equivalently, all cycles of $L^g$ are nonnegative. 
\end{theorem}

Suppose $L^g \in \R^{\img \times \img}$ has min-plus eigenvalue zero. The \emph{Kleene star} of $L^g$ is the matrix
$$(L^g)^\ast = I \underline{\oplus} \Bigg(\underline{\bigoplus}_{i=1}^{\cardg} (L^g)^{\underline{\odot} i}\Bigg),$$ where $I$ denotes the min-plus identity matrix with zero on the diagonal and $+\infty$ on the off-diagonal. Kleene stars are also known as distance matrices \cite{murota2003discrete}, or strong transitive closures \cite[\S 1.6.2.1]{butkovivc2010max}. Computing the Kleene star is equivalent to computing the all-pairs shortest path in the graph with edge weights $L^g$, which can be cast as a linear program \cite{ahuja1993network}. Classical results in tropical linear algebra \cite[Thm. 5.1.3]{maclagan2015introduction}, \cite[Thm 4.2.4]{butkovivc2010max} immediately translate into the following algebraic characterization of IC payments.
\begin{definition}
Let $g:T\to [m]$ be an IC outcome function and $L^g \in \R^{\img \times \img}$ its allocation matrix. The set of \emph{tropical generating payments} $V(\mathcal{P}(g))$ of $\mathcal{P}(g)$ is the set of columns of $(L^g)^*$ viewed as points in $\mathbb{TA}^{\cardg-1}$.
\end{definition}
\begin{theorem}\label{thm:tropicaleigenspacegenerator}
Let $g:T\to [m]$ be an IC outcome function. Each $p \in \mathcal{P}(g)$ can be written as a min-plus linear combination of vectors in $V(\mathcal{P}(g))$, that is,
$$ p = \underline{\bigoplus}_{v \in V(\mathcal{P}(g))} a_v \odot v $$ 
for some $a_v \in \R$. Furthermore, the cardinality of $V(\mathcal{P}(g))$ equals the dimension of $\mathcal{P}(g)$. 
\end{theorem}
A brief discussion is in order. The theorem states that a price is IC if and only if it can be written as a \emph{min-plus} linear combination of the tropical generating payments. Indeed, adding constants to payments does not violate IC, neither does taking point-wise minima of IC payments. The surprising part is that \emph{all} payments are obtainable this way. Note that this is only true using tropical algebra. In fact, tropical generating payments are ordinary vertices of $\mathcal{P}(g)$, but the converse is not true. The number of ordinary vertices of $\mathcal{P}(g)$ ranges between $1$ and $\binom{2\cardg-2}{\cardg-1}$ \cite[Prop. 19]{develin2004tropical}, while the number of tropical generating payments equals the dimension which is always between $1$ and $\cardg$. Thus the set of IC payments $\mathcal P(g)$ thus forms a \emph{tropical simplex} spanned by the generating payments.

The $i$-th column of $(L^g)^*$ consists of the shortest paths from each of the $d$ nodes towards node $i$, that is
\begin{equation}\label{eqn:geometric.kleene.star} ( L^g)^*_{ij}=\sup\{p_i-p_j~:~p\in\mathcal{P}(g)\}, \qquad \text{for all}~i,j\in \img.
\end{equation}
In economics terms, the entries of the Kleene star are the maximal payment differentials between outcomes consistent with IC of $g$, by this we mean the following. Fixing a price for outcome $i$, the number $(L^g_{ij})^*$ is the maximal payment difference the mechanism can charge between outcome $i$ and $j$. If the difference were higher, then some type in $g^{-1}(j)$ would find it optimal to deviate to some outcome $k\neq j$, thereby violating IC.

\input{TGMDfigure_payments.tex}

The set of generating payments $V(\mathcal{P}(g))$ is given by the columns of the Kleene star $(L^g)^\ast$ viewed as vectors in $\mathbb{TA}^{\cardg-1}$. While the Kleene star $(L^g)^\ast$ has $\cardg$ columns, two columns which differ by a constant entry-wise are the same point in $\mathbb{TA}^{\cardg-1}$. Hence, the cardinality of $V(\mathcal{P}(g))$ is between $1$ and $\cardg$. We now give a geometric characterization of the number of generating payments.
\begin{definition}\label{defn:dimension}
For $p \in \TA^{m-1}$, the \emph{dimension graph} of $p$ with respect to $T$ is the graph on $m$ nodes, with edge $\{i,j\}$ if and only if
\begin{equation}\label{eqn:dimension.graph} T\cap \maxH_{ij}(-p)\neq \emptyset.\end{equation}  
\end{definition}

Observe that for each pair $\{i,j\}$, the sector $\maxH_{ij}(-p)$ is a translation of the sector $\maxH_{ij}(0)$ by~$p$. Hence calculating the dimension graph of a point allows for a simple geometric algorithm by considering translations of the sets $\maxH_{ij}(0)$ and checking whether they intersect $T$, or conversely, translating $T$ and checking whether it intersects $\maxH_{ij}(0)$. 

\begin{theorem}\label{thm:dimension}
Let $g:T\to[m]$ be an IC outcome function and $p$ be a point in the relative interior of $\Plift(g)$. Then the number of generating payments equals the number of connected components in the dimension graph of $p$ with respect to $T$.
\end{theorem}

\input{TGMDfigure_dimensiongraph.tex}

By Theorem \ref{thm:tropicaleigenspacegenerator} the generating payments $V(\mathcal{P}(g))$ are the unique generators of the tropical eigenspace of $L^g$. Theorem \ref{thm:dimension} characterizes their number in terms of a geometric property of the cell $\Plift(g)$.  The representation in terms of generators is dual to the hyperplane description of the price simplex from Proposition \ref{prop:projection.payments}, \cf Remark \ref{rem:duality}. In other words $V(\mathcal{P}(g))$ is also the set of tropical vertices of $\mathcal{P}(g)$. Yet it is the geometric aspect manifested in Theorem~\ref{thm:dimension} and the calculation in (\ref{eqn:geometric.kleene.star}) that makes the characterization of Theorem \ref{thm:tropicaleigenspacegenerator} appealing, for it does not require reference to $L^g$ but instead provides a direct link to the geometry of $T$ via $\tconv(T)$. The tropical apparatus is uniquely suited for this task by bridging algebraic and geometric techniques. Indeed, the link between the number of algebraic generators and geometric dimension parallels classical linear algebra. It shows how the cells of the tropical polytope $\tconv(T)$ can be used to study mechanisms.  For instance, in Section \ref{sec:re.details}  we derive a novel geometric characterization of revenue equivalence.

\begin{remark}
We offer a remark on Theorems \ref{thm:tropicaleigenspacegenerator} and \ref{thm:dimension}. One is stated in terms of $\mathcal{P}(g)$, the potentially lower dimensional payment set defined in (\ref{eqn:pg}), while the other is stated in terms of $\Plift(g)$, the embedding of the payment set into the tropical convex hull $\tconv(T)$, defined in (\ref{eqn:pg.tropical}). By Proposition \ref{prop:projection.payments}, the cells $\mathcal P(g)\subset \R^\img$ and $\Plift(g)\subset \R^m$ are isomorphic as tropical simplices, so that their dimension and the number of tropical vertices agrees. The calculation in (\ref{eqn:geometric.kleene.star}) can be carried out for all $i,j\in [m]$ on $\Plift(g)$, thereby lifting the tropical generating payments from $\R^\img$ to $\R^m$ to obtain a $V$-representation of $\Plift(g)$. In other words, the two theorems are fully compatible. In examples, for visualization purposes, it is often more convenient to work with $\Plift(g)$. In algebraic descriptions, as in Theorem \ref{thm:tropicaleigenspacegenerator}, it is more natural to work with $\mathcal{P}(g)$.\
\end{remark}

\input{TGMDfigure_dimensiongraph2.tex}

\subsection{Generating payments and revenue equivalence}\label{sec:re.details} 

There has been an effort to obtain a general characterization of the set of IC payments, see \cite{kos2013extremal,carbajal2013mechanism}. 
Kos and Messner \cite[Theorem 1]{kos2013extremal} showed that the supporting payments of an IC outcome function $g$ satisfy type-wise upper and lower bounds, which they termed  ``extremal payments". However, while extremal payments themselves are IC prices, boundedness by extremal payments is generally not sufficient for a payment rule to be IC. In contrast, generating payments completely characterize IC prices. Moreover, Theorems \ref{thm:tropicaleigenspacegenerator} and \ref{thm:dimension} show that they provide a minimal characterization and that their number is related to the geometry of the type space. 

Extremal payments are defined in terms of certain path-weights in the type graph of $g$, obtained by refining the results of \cite{rochet1987necessary, heydenreich2009characterization}, and extending \cite[Prop. 2]{chung2007non-differentiable}. Geometrically they correspond to two  extremal points in $\Plift(g)$. Indeed, the upper and lower bounds are the coordinate-wise maximal, respectively minimal points in $\Plift(g)\in\tconv(T)$. Examples \ref{ex:information.rents.2} and \ref{ex:information.rents} detail this construction showing how obtain an equivalent definition that is path weight free. It was argued in \cite[Sec. 7]{kos2013extremal} that such a definition is desirable. Moreover they show how our geometric tools naturally lend themselves to the economic analysis of equilibrium utilities and the distribution of information rents. 

If the upper and lower bounds on IC payments coincide, then an IC outcome function is said to be \emph{revenue equivalent} (RE). In other words, $g$ is RE, if its set of incentive compatible payments $\mathcal P(g)$ consists of exactly one point up to tropical scalar multiplication. 
A type space is revenue equivalent if any IC outcome function  defined on it is RE. The literature on revenue equivalence is vast \cite{myerson1981optimal, milgrom2002envelope, krishna2001convex, jehiel1996how, vohra2011mechanism, heydenreich2009characterization, chung2007non-differentiable}, see \cite{chung2007non-differentiable, vohra2011mechanism} for a comprehensive review. As a direct consequence of Theorem \ref{thm:dimension} and Proposition \ref{prop:projection.payments}, we obtain the following new characterization of revenue equivalence relating geometry and algebra.

\begin{corollary}\label{cor:re}
Let $g$ be an IC outcome function. Then $g$ is revenue equivalent, if and only if the dimension graph of $p\in \Plift(g)$ is  connected, which is the case if and only if the dimension of $\mathcal P(g)$ is zero. 
\end{corollary}

This characterization provides a unified framework for our study of both, the revenue equivalence of outcome functions, and that of type spaces. The key point is that it treats RE as a property of cells of $\tconv(T)$. A benefit of this view is that it provides an explicit characterization of RE rooted in the geometry of the type space. Another benefit is that the characterization extends to collections of cells, and thus to subsets of the set of all IC mechanisms. The following example demonstrates these points. 

\input{TGMDfigure_re2.tex}

Two corollaries of Theorem \ref{thm:dimension}, related to Corollary \ref{cor:re} should not go unmentioned, as they connect our results to the existing literature. 

\begin{corollary}\label{cor:re.function} 
Let $g:T\to [m]$ be IC and let $L^g$ be its allocation matrix. For $p$ in the relative interior of $\Plift(g)$ the following are equivalent.
	 \begin{enumerate}
	 \item The dimension graph of $p$ is connected.
	 \item The outcome function $g$ is revenue equivalent.
	 \item The matrix $(L^g)^\ast$ is skew-symmetric.
	 \end{enumerate}
\end{corollary}

\begin{corollary}\label{cor:re.type} 
Let $T \subset \R^{m}$ be a type space. The following are equivalent. \begin{enumerate}
 \item 	For each $p~\in~\tconv(T)$ the graph of $p$ is connected.
 \item The type space $T$ is revenue equivalent.
 \end{enumerate}
\end{corollary}

The equivalence of parts (2) and (3) in Corollary \ref{cor:re.function} is the result obtained by Heydenreich, M\"uller, Uetz and Vohra \cite[Theorem 1]{heydenreich2009characterization}, for finite outcome functions. 
Since the Kleene star $( L^g)^*$ records the length of shortest paths, its skew symmetry is equivalent to the anti-symmetry of the graph distance defined in \cite{heydenreich2009characterization}.  We augment this characterization by the geometric condition (1), which can be verified on the min-plus arrangement $\minH(-T)$. It makes the study of the RE property explicit and does not necessitate passing to induced allocation matrices. 

Corollary~\ref{cor:re.type} gives an alternative characterization for RE type spaces. A recent result in this vein is due to Chung and Olszewski \cite[Theorem 4]{chung2007non-differentiable}, which states that $T$ is RE if and only if there do not exist disjoint subsets $B_1, B_2$, a function $r: B_1 \cup B_2\to \R$, and an $\epsilon >0$, such that $T$ equals the union of two non-empty sets $V_+(B_1,\epsilon,r)$ and $V_-(B_2,\epsilon,r)$, whose definitions depend on the parameters given. Given a specific type space $T$, it is not immediate how to apply this theorem to check whether $T$ is RE, while checking whether a type space is RE using Corollary~\ref{cor:re.type} via the geometric algorithm described following the definition of the dimension graph is surprisingly simple. Example \ref{ex:re} demonstrates this.

As pointed out in \cite{heydenreich2009characterization}, results which establish RE of a type space have limited applicability, as there are type spaces on which RE and non-RE outcome functions coexist. Corollary \ref{cor:re} is immune to this scrutiny. The emphasis on cells allows for the analysis of type spaces in light of geometry. It combines the analytic detail of a characterization of the RE property on the level of outcome functions from Heydenreich, M\"uller, Uetz and Vohra, with the appealing features of a geometric criterion emphasized in the work of Chung and Olszewski. For instance, Example \ref{ex:re} showed how Corollary \ref{cor:re} can be used to study on which parts of the type space RE holds and where it fails, thereby characterizing the subsets of IC mechanisms which are RE and non-RE.

We conclude this section with two examples aiming to build economic intuition for our geometric theory. The first provides a geometric construction of the extremal payments of \cite{kos2013extremal}, the second shows that our theory can be used to understand how the geometry of cells interplays with the ability of the mechanism designer to distribute information rents in the absence of RE.

\input{TGMDfigure_informationrents.tex}

\input{TGMDfigure_informationrents2.tex}

\section{Characterization of Feasible Allocation Matrices}\label{sec:allocation.matrices}
The allocation matrix defined in (\ref{eqn:allocation.matrix}) plays a prominent role in the verification of IC and RE. It encodes important properties of an outcome function, such as weak monotonicity, IC, supporting payments and RE. Yet characterizations based on it are implicit by nature of its definition. To remedy this, two essential hurdles need to be overcome. Firstly, given a type set $T$ and matrix a $L$, we need a criterion to decide whether $L$ is induced by some outcome function $g:T\to [m]$. In other words, we must characterize the image $g\mapsto L^g$. We shall call such matrices \emph{realizable}. Secondly, if $L$ is realizable, a method to construct the functions that realize it is desirable. Theorem \ref{thm:realizable} resolves these two questions. It thus moves the focus from studying outcome functions to analyzing specific matrices.

For ease of exposition, we shall assume that $g: T \to [m]$ is an onto outcome function. This is without loss, as any outcome function is onto upon projecting the type space. We note however that the theorem can be adapted to the general case at the expense of additional notation.

\subsection{Realizability of allocation matrices}\label{sec:realizability}

Let $L$ be a $m\times m$ matrix with zero diagonal and write $L_j$ for its $j$'th column. To economize on notation we shall write $\maxL_j := \maxH_j(L_j)$, $\minL_j:=\minH_j(L_j)$ and $\maxL_{jk} := \maxH_{jk}(L_j)$.
For sets $A, B \subset \R^m$, write $\mathsf{d}(A,B)$ for the infimum of the euclidean distance between pairs of points in these sets. 

\begin{definition}\label{defn:seperating}
For $j,k \in [m], j \neq k$, define $\mathcal{I}_{jk} = \mathcal{\overline{L}}_{jk}\cap T$. A \emph{$(j,k)$-witness} is a sequence  
$\{s^{j,r}: r \geq 1\} \subseteq T \cap \maxL_j^\circ$ such that
$$ \lim_{r \to \infty} \mathsf{d}(s^{j,r}, \mathcal{\overline{L}}_{jk} )=0.$$
We say that \emph{$L$ separates $T$ at $(j,k)$} if
\begin{equation}\label{eqn:separate.t}
\mathsf{d}(T \cap \mathcal{\overline{L}}_j, \mathcal{\overline{L}}_{jk})=0,
\end{equation}
and in addition, whenever $\mathcal{I}_{jk} = \mathcal{I}_{kj} = \{s\}$ for some $s \in \mathbb{TA}^{m-1}$, then there exists a $(j,k)$-witness or a $(k,j)$-witness. 
Say that \emph{$L$ separates $T$} if $L$ separates $T$ for all $j,k \in [m], j \neq k$. 
\end{definition}

Given a matrix $L \in \R^{m \times m}$, we will say that \emph{$L$ is realized by $g$} if $L$ is the allocation matrix $L^g$ for some outcome function $g: T \to [m]$. We shall call $L$ \emph{realizable} if it is realized by some~$g$. 

\begin{theorem}\label{thm:realizable}
Fix a type space $T \subset \R^m$. Let $L \in \R^{m \times m}$ be a matrix with zero diagonal. Then $L$ is realizable if and only if $L$ separates $T$, and 
$$T \subseteq \bigcup_{k=1}^m \overline{\mathcal{L}}_k.$$
\end{theorem}

Theorem \ref{thm:realizable} demonstrates how network flow problems can be studied geometrically. As discussed in the introduction, one issue with the characterization of IC based on cycle lengths is the lack of an inverse construction. The proof of Theorem \ref{thm:realizable} is constructive: it yields an outcome function $g$ such that $L = L^g$.

\begin{corollary}\label{cor:L.payment}
Let $L = L^g \in \R^{m \times m}$ be the allocation matrix of some (not necessarily IC) outcome function $g: T \to [m]$. Then \begin{equation}\label{eqn:geometric.eigenspace}\mathcal P(g) = \bigcap_{k\in 1}^m\minL_k.\end{equation} In particular, $g$ is IC if and only if this intersection is non-empty.
\end{corollary}

\input{TGMDfigure_allocationmatrix.tex}

\input{TGMDfigure_realize.tex}

\subsection{Weak Monotonicity}\label{sec:weak.monotone}
Recall that an outcome function $g$ is said to be \emph{weakly monotone} or two-cycle monotone if the matrix $L^g + (L^g)^\top$ is element-wise nonnegative. Using tropical geometry, this algebraic condition on $L^g$ is translated into the following geometric criterion on the row vectors of $L^g$ and the min-plus arrangement $\minH(-T)$.

\begin{proposition}\label{prop:weak.mon}
Fix a type space $T \subset \R^m$. Let $L = L^g \in \R^{m \times m}$ be the allocation matrix of some (not necessarily IC) outcome function $g: T \to [m]$. Then $g$ is weakly monotone if and only if the open sets $\overline{\mathcal{L}}^\circ_1, \ldots, \overline{\mathcal{L}}^\circ_m$ are pairwise disjoint.  
\end{proposition}

It is known that weak monotonicity is generally not sufficient for IC, inspiring a still evolving literature identifying conditions on the type space when it is, see \cite{saks2005weak,bikhchandani2006weak,ashlagi2010monotonicity,archer2008characterizing,jehiel1996how,kushnir}. Together, Corollary \ref{cor:L.payment} and Proposition \ref{prop:weak.mon} indicate geometric constraints on the type space for weak monotonicity to imply cyclical monotonicity. More generally, any weakening of cyclical monotonicity requires finding conditions under which a non-empty intersection of subsets of the sectors $\minL_i$ already implies non-emptiness of the intersection of all sectors $\minL_i$.

\input{TGMDfigure_wm.tex}

\section{Extensions}\label{sec:extensions}

\subsection{Multiplayer mechanisms}\label{sec:multiplayer}

We give the generalization of Theorem \ref{thm:main.generic} to multiplayer, dominant strategy incentive compatible mechanisms.
In this case, the matrix that represents the outcome function $g$ and the covector of a cell of $\minH(-T)$ are generalized to tensors. Let us elaborate. Suppose that there are $n\in \mathbb{N}$ players, each player $i\in [n]$  has type space $T_i\subset \R^{m}$. A multiplayer mechanism is a pair $(g,p)$, with outcome function $g: T_1 \times \dots \times T_n \to [m]$ and payment $p: T_1 \times \dots \times T_n \to \R^{m\times n}$. Say that $(g,p)$ is dominant strategy incentive compatible (D-IC) if for each player $i \in [n]$, for each possible announcement $t_{-i} \in \prod_{i\neq j} T_j$ of the other players, the induced mechanism for player $i$ is IC. View $g$ as a tensor of dimension $m\times T_1\times \cdots\times T_n$ having entries in $\{0,1\}$, where $g_{k,t_1, \ldots t_n}=1$ if and only if $g(t_1, \ldots, t_n)=k$. Define a \emph{multiple covector} $\nu$ to be a tensor of dimension $m\times T_1\times \cdots \times T_n$ with entries in $\{0,1\}$ such that for each $\ti\in \Ti$, the axis-aligned sub-matrix $\nu_{\bullet, \bullet, \ti}$ of dimension $m\times T_i$ is a covector of $\minH(-T_i)$. The following is immediate from the definition of D-IC for multiplayer outcomes.

\begin{proposition}\label{prop:multiplayer}
The outcome function $g$ is incentive compatible if and only if for each $i \in [n]$, every axis-aligned submatrix of dimension $m\times T_i$ is contained in the covector of some cell in $\tconv(T_i)$. 
\end{proposition}

\input{TGMDfigure_multiplayer.tex}

\subsection{General type spaces}\label{sec:general.types}
Thus far the main results of our paper, Theorems \ref{thm:main.generic}, \ref{thm:dimension} and \ref{thm:realizable}
have been stated for compact type spaces $T$. Some results also require $T$ to be tropically generic (cf. Corollary \ref{cor:cell.mechanism.bijection}).
In this section we demonstrate how these results can be generalized through examples. 
In particular, Examples \ref{ex:compact.closure} and \ref{ex:noncompact.closure} show the pathologies that can occur with non-compact or non-generic type spaces, while Examples \ref{ex:limit} and \ref{ex:generic.perturbation} show how to extend the theory to these cases via generic perturbations and finite approximations. It is straight forward to formalize the ideas in these examples using standard techniques.

One can directly adapt the definitions to apply to non-compact type spaces as follows.

\input{TGMDfigure_noncompact.tex}

We remark, that defining the tropical polytope on $T$ as in (ii) of Example \ref{ex:compact.closure}, one can apply our theory by considering only outcome functions with finite prices, and considering projections to obtain any  outcome function that is not onto, thus not having to explicitly handle the subtleties posed by allowing for infinite prices.  

This case can also be handled using generic perturbations. 

\input{TGMDfigure_limit.tex}
\input{TGMDfigure_generic.tex}

We remark that by combining the two proceed examples one can choose finite and generic sequences of points that approximate any given type space and apply the present theory to these approximations. Thus, while there no longer is a bijection between onto outcome functions and  full-dimensional cells of the tropical polytope in the non-generic case, what remains true is that a cell supports an onto outcome function if and only it it is obtainable as a limit of full-dimensional bounded cells of finite generic approximations. Generic perturbations are a standard technique in tropical geometry, see \cite{maclagan2015introduction}.

\section{Proofs}\label{sec:proofs}

\subsection{Proofs of results in Section \ref{sec:overview}}
\begin{proof}[Proof of Proposition \ref{prop:covector}]
Fix $p \in \R^m$ and pick $g \in \mathcal{G}(p)$. Since $(g,p)$ is IC, we have for $g(t) = i$, the inequality
\begin{equation}\label{eqn:one.edge}
t_i-p_i \leq t_j-p_j \quad \forall j \in [m], j \neq i.
\end{equation}
Therefore, $p \in \minH_i(-t)$ by definition of $\minH_i(-t)$ so that $\covec(p)$ contains the edge $(t,i)$. This shows that $g$ is a subgraph of $\covec(p)$. Conversely, suppose the graph of $g$ is a subgraph of $\covec(p)$. By hypothesis $g$ is a well-defined outcome function, so each $t \in T$ is assigned to some $g(t) \in [m]$. Choose $t \in T$ and suppose $g(t) = i$. Since $p \in \minH_i(-t)$ the relation (\ref{eqn:one.edge}) holds and $(g,p)$ is IC. Thus $g \in \mathcal{G}(p)$, as claimed. 
\end{proof}

We will make repeated use of Corollary 12 in \cite{develin2004tropical}, which we state here for ease of reference.

\begin{lemma}[{\cite[Corollary 12]{develin2004tropical}}]
A cell $\sigma\in \minH(-T)$ is bounded modulo addition of constants if and only if for all $i\in [m]$ there is $t\in T$ such that $\sigma\subset \minH_i(-t)$.	
\end{lemma}

\begin{proof}[Proof of Proposition \ref{prop:projection.payments}]
If $\img = [m]$, then $\pi_g$ is a the identity function, and the statement is trivial. Suppose $\img$ is a strict subset of $[m]$. If $g$ is not IC, then $\Plift(g) = \mathcal{P}(g) = \emptyset$, and the statement is again trivial. Thus, we shall assume that $g$ is IC and $\img$ is a strict subset of  $[m]$.
Consider the polyhedron $\sigma =\bigcap_{t\in T}\minH_{g(t)}(-t)$. By IC $\sigma$ is non-empty and thus a cell in $\minH(-T)$. However, since $g$ is not onto, $\sigma$ is unbounded modulo addition of constants by the foregoing lemma. Thus, by definition, $\Plift(g)\in\tconv(T)$ is obtained by restricting $\sigma$ to $\tconv(T)$. Again by the foregoing lemma, for each $j\in [m]\backslash \img$ there must be some $s\in T$ such that $\Plift(g)\subset \minH_j(-s)$. If we define $g(s)=i_j$, if follows that $\Plift(g)\subset \minH_j(-s)\cap \minH_{i_j}(-s)$, i.e. $q_{i_j}-q_j=c^j$ with $c^j= s_{i_j}-s_j$ is valid for all $q\in \Plift(g)$. This also shows that the constants $c^j$ cannot depend on the choice of $s$. Now suppose $p\in \mathcal{P}(g)$, lift it to $\bar p\in \R^m$ by setting $\bar p_i= p_i$ if $i\in \img$, and $$\bar p_j = p_{i_j} -c^j $$ if $j\in [m]\backslash \img$. The definition of the lift shows that $\bar p\in \Plift(g)$. We now claim that the lift is an affine isomorphism. It is clear that the map is affine. Obviously, if $q\in \mathcal{P}(g)$, then $\pi_g(\bar p)=p$. To conclude, we claim that if $q\in \Plift(g)$, then the lift of $\pi_g(q)$ is again $q$. But this is immediate, since all $q\in \Plift(g)$ satisfy the linear relation that was used to construct the lift.
\end{proof}

\begin{proof}[Proof of Proposition \ref{prop:covec.full.dim}]
First suppose $\sigma\in \minH(-T)$ is a full-dimensional cell and let $\nu$ denote its covector. Pick $p\in \operatorfont{int}(\sigma)$, by definition $p\in \bigcap_{(i,t)\in \nu}\operatorfont{int}(\minH_i(-t))$, i.e. $p_i-t_i< p_j-t_j$ for all $(i,t)\in \nu$ and $j\neq i$. This shows that for each $t\in T$ there is a unique index $i$ that attains the minimum $j\mapsto p_j-t_j$. In other words, for each $t\in T$ there is a unique index $i\in [m]$ such that $\nu_{it}=1$ and $\nu_{jt}=0$ for $j\neq i$. Defining $g(t)=i$ if and only if $\nu_{it}=1$ gives a well-defined and IC outcome-function, so that (\ref{eqn:g.sigma}) is the only outcome function that is a subgraph of $\covec(\sigma)$. For the converse, suppose that $\sigma$ is not full-dimensional. As $\minH(-T)$ is a polyhedral complex, $\sigma$ equals the intersection of at least two full-dimensional cells $\nu_1, \nu_2$ of $\minH(-T)$. If $f_1$ is the unique IC outcome associated to $\nu_1$, and $f_2$ is the unique IC outcome associated to $\nu_2$, then $f_1, f_2$ are both subgraphs of $\covec(\sigma)$. Therefore, $\sigma$ has more than one IC outcome function associated with it.
\end{proof}

\subsection{Proofs of results in Section \ref{sec:main}}

\begin{proof}[Proof of Theorem \ref{thm:main.generic}]
	If $\sigma$ is a full-dimensional cell, then the claim  follows from Proposition \ref{prop:covec.full.dim}. For the converse direction, suppose $g$ is an IC outcome function. Then by IC, the polyhedron $$\sigma= \bigcap_{t\in T} \minH_{g(t)}(-t)$$is non-empty and a cell of $\minH(-T)$. We now claim that $\sigma$ is full-dimensional. If not, then there are $i\neq j$ and a $c\in \R$ such that $p_i-p_j=c$ is valid for all $p\in \sigma$, equivalently $\sigma\subset \minH_{ij}(-t)$ for some $t\in T$. Since $g$ is a function, there must be $t\neq t'$ with $t_i-t_j=t'_i-t'_j=c$. Projecting $T$ onto coordinates $i$ and $j$, we found $t$ and $t'$ on a tropical hyperplane contradicting genericity of $T$. Hence $\sigma$ must be full-dimensional. 
	The final claim is clear from the definition of $\Plift(g)$. 
\end{proof}

\begin{proof}[Proof of Corollary \ref{cor:cell.mechanism.bijection}]This follows from Theorem \ref{thm:main.generic} and  \cite[Corollary 12]{develin2004tropical}.	
\end{proof}

The following is an adaptation of \cite[Prop. 17]{develin2004tropical}.
\begin{lemma}\label{lem:projection}
Let $\sigma \in \minH(-T)$ be a cell, and $p$ a point in its relative interior. The number of connected components of the dimension graph of $p$ equals the dimension of $\sigma$ as a polyhedron in $\R^m$.
\end{lemma}

\begin{proof}  
Let $\sigma$ be a cell in $\minH(-T)$ and let $\nu$ denote its covector. Then
$$ \sigma = \bigcap_{t \in T} \minH_{i \in \nu(t)}(-t). $$
Note that the dimension graph is completely determined by the covector namely, $$ \{i,j \}~\text{is an edge}~\Leftrightarrow ~\maxH_{ij}(-p)\cap T\neq \emptyset ~\Leftrightarrow~ \text{there is}~t\in T~\text{with}~\nu_{it}=\nu_{jt}=1.$$

We will use induction on the dimension. First suppose that the dimension of $\sigma$ is $m$, the maximal possible. This means that the interior of $\sigma$ is non-empty. Pick $p\in \operatorfont{int}(\sigma)$, by definition there is an $\epsilon>0$ such that such that $t_i-p_i>t_j-p_j+\epsilon$ for all $(i,t)\in \nu$ and all $i\neq j$. Clearly this implies that $\maxH_{ij}(-p)\cap T =\emptyset$ and the dimension graph is totally disconnected. For the converse direction, assume the dimension graph of $p$ is totally disconnected. Fix $i\in [m]$, then by the compactness of $T\cap \maxH_i(p)$ we find an $\epsilon_i>0$ such that $$t_j-p_j \geq t_i-p_i +\epsilon_i $$ for all $t\in T$ such that $(i,t)\in \nu$. Letting $\epsilon$ be the minimum of $\epsilon_1, \ldots, \epsilon_m$, we see that $p$ contains an open $\epsilon$-neighborhood in $\sigma$, so that $\sigma$ is full-dimensional.

Now suppose there is a $t\in T$ with $\sigma \subset \minH_{ij}(-t)$. This is the case if and only if for some $t\in T$ and $c:=t_i-t_j$, $p_i-p_j=c$ is valid for all $p\in \sigma$. Projecting out coordinate $j$ we obtain $\sigma'\subset \R^{m-1}$ which is affinely isomorphic to $\sigma$, so that both have the same dimension. 
Consider now the tropically linear map $\TA^{m-1}\to \TA^{m-2}$ defined by $x'_k =x_k$ for $k\neq i,j$ and $x'_i=\max\{x_i, x_j+c \}$. Denote the image of $T$ by $T'$ and the image of $p\in \sigma$ by $p'$. We now claim that $\sigma'$ is a cell in $\minH(-T')$ which is tropically isomorphic to $\sigma$. Indeed, for $(i,t)\in \nu$, we have $t_i= t_j- (p_j-p_i)= t_j-c$. So $t'_i =t_i$ if $(i,t)\in \nu$. Thus $\nu'_{kt'}=\nu_{kt}$ for $k\neq i,j$ and $\nu'_{it'}=1$ if and only if $\nu_{it}=1$ or $\nu_{jt}=1$ is the covector of $p'$ with respect to $T'$. Thus $\nu'$ is the covector of $\sigma'$ with respect to $T'$. The dimension graph of $p'\in \sigma'$ with respect to $T'$ is obtained by contracting the edge $\{i,j\}$. Now induct on the dimension.  
\end{proof}

\begin{proof}[Proof of Theorem \ref{thm:dimension}]
Let $g:T\to[m]$ be an IC outcome function.
By Lemma \ref{lem:projection}, the dimension of $\Plift(g)$ equals the number of connected components of its dimension graph. Since $\Plift(g)$ is a cell, it is a tropical simplex whose generators equals the dimension in $\R^m$, \cite[Theorem 7]{joswig2010tropical}. Hence also its tropical dimension equals the number of connected components of its dimension graph. 
\end{proof}

\begin{proof}[Proof of Corollary \ref{cor:re.type}]
If $T$ is RE, then for every $p \in \tconv(T)$, its dimension graph equals that of some RE mechanism $(g,p)$ by Proposition \ref{prop:covector} which is connected by hypothesis. Conversely, suppose the dimension graph of every $p$ in $\tconv(T)$ is connected. Take a IC mechanism $(g,p)$, with $p$ in the relative interior of $\mathcal{P}(g)$. As its dimension graph is connected, $\Eig(L^g) = \{p\}$, so $(g,p)$ is RE. 
\end{proof}

\subsection{Proofs of results in Section \ref{sec:allocation.matrices}}

\begin{proof}[Proof of Theorem \ref{thm:realizable}]
Suppose $L$ is realized by $g$. For each $j = 1, 2, \ldots, m$, $g^{-1}(j) \subseteq \overline{\mathcal{L}}_j$, since $\max_{k \in [m]} (L_{jk} + t_k) = L_j \overline{\odot} t =  t_j$. The sets $g^{-1}(1), \ldots, g^{-1}(m)$ partition $T$, so $T \subseteq \bigcup_{j=1}^m \overline{\mathcal{L}}_j$. It remains to show that $L$ separates $T$ at an arbitrary pair $\{j,k\}$, where $j, k \in [m], j\neq k$. We have that $\mathsf{d}(T \cap \mathcal{\overline{L}}_j, \mathcal{\overline{L}}_{jk})=0$ since $L_{jk} = \inf_{t \in g^{-1}(j)}\{t_j - t_k\}$. If $\mathcal{I}_{jk} = \mathcal{I}_{kj} = \{s\}$, then if $g(s) = j$, the infimum in the definition of $L_{kj}$ must be achieved by a $(k,j)$-witness, so there must exist a $(k,j)$-witness. Conversely, if $g(s) = k$, then there must exist a $(j,k)$-witness. This shows that $L$ separates $T$ at $\{j,k\}$, as desired. For the converse direction, suppose $L$ is a matrix with zero diagonal such that $T \subseteq \bigcup_{j=1}^m \overline{\mathcal{L}}_j$ and 
$L$ separates $T$. Define $g: T \to [m]$ as follows. For a point $t \in T \cap \overline{\mathcal{L}}_j^\circ$, let $g(t) = j.$ The remaining points must lie on $\bigcup_{j,k \in [m], j \neq k}\mathcal{I}_{jk}$ by definition of $\mathcal{I}_{jk}$'s. Assign these points such that points on $\mathcal{I}_{jk}$ have either outcome $j$ or $k$, and such that on every non-empty boundary $\mathcal{I}_{jk}$ there exists a point with outcome $j$. The only case where this cannot be done is 
if $\mathcal{I}_{jk} = \mathcal{I}_{kj} = \{s\}$. In this case, if there is a $(j,k)$-witness, set $g(s) = k$, else, since $L$ separates $T$, there must exists a $(k,j)$-witness, so set $g(s) = j$. We claim that $L$ is realized by $g$. Fix $j, k \in [m], j\neq k$. By definition of $g^{-1}(j)$,
\begin{equation}\label{eqn:infimum.ell}
\inf_{s \in g^{-1}(j)} (s_j - s_k) \geq L_{jk}.
\end{equation}
Furthermore, there must be a point in $\mathcal{I}_{jk}$ or a $(j,k)$-witness in $g^{-1}(j)$. So (\ref{eqn:infimum.ell}) holds with an equality, thus $L = L^g$, as claimed.
\end{proof}

\begin{proof}[Proof of Proposition \ref{prop:weak.mon}]
Fix a pair of indices $j,k \in [m]$, $j \neq k$. We claim that
\begin{align*}
L_{jk} + L_{kj} > 0 &\Leftrightarrow \maxL_j \cap \maxL_k = \emptyset, \\
L_{jk} + L_{kj} = 0 &\Leftrightarrow \maxL_j \cap \maxL_k \mbox{ on their boundaries, i.e. } \maxL_j \cap \maxL_k = \partial\maxL_j \cap \partial\maxL_k \\
L_{jk} + L_{kj} < 0 & \Leftrightarrow \maxL_j \cap \maxL_k \mbox{ in their interiors, i.e. }\maxL^\circ_j \cap \maxL^\circ_k \neq \emptyset.
\end{align*}
This claim implies statement (4) by definition of weakly monotone. Now let us prove the claim. Note that $\maxL_j$ and $\maxL_k$ are closed polyhedra, so they either do not intersect, intersect on their boundaries, or intersect in their interiors. So it is sufficient to prove the last two equivalences in the statements. Suppose there exists $t \in \maxL_j \cap \maxL_k$. Then by definition of the sectors,
$$ L_{jk} + L_{kj} \leq (t_j - t_k) + (t_k - t_j) = 0. $$
In particular, strict equality holds if and only if $t$ lies on the boundary of $\maxL_j \cap \maxL_k$, while strict inequality holds if and only if $t$ lies in either of their interiors.  In that case, one can pick a $t' \in \maxL^\circ_j \cap \maxL^\circ_k$, then $L_{jk} + L_{kj} < (t'_j - t'_k) + (t'_k - t'_j) = 0$. This proves the desired statement. 
\end{proof}

\section{Acknowledgements}
We wish to thank John Weymark for suggesting to explore possible connections of mechanism design to tropical geometry. Our geometric approach is inspired by his work with coauthors in \cite{cuff2012dominant}.
 Part of this work was written during the Hausdorff Center Conference ``Tropical Geometry and Economics''. We thank the participants for helpful discussions. The Hausdorff Center's role in catalyzing research in mathematical economics is gratefully acknowledged. We also wish to thank Benny Moldovanu and Martin Pollrich for helpful discussions. Finally we thank Paul Edelman and John Weymark for pointing out an error in an earlier version, and Alexey Kushnir for generously sharing his manuscript.

\section{Conclusion}\label{sec:conclusion}
In this paper we developed a geometric theory of incentive compatibility that naturally reflects economic incentives, Section \ref{sec:connections} and Theorem \ref{thm:main.generic}. Our approach uses tropical geometry, which combines algebraic and geometric techniques, providing a framework for mechanism design not utilized prior. Importantly, it is this interplay of techniques which allowed us to give a systematic characterization of all IC prices of an outcome function, Theorems~\ref{thm:tropicaleigenspacegenerator} and \ref{thm:dimension}, and let us show how to use cyclical monotonicity constructively, Theorem \ref{thm:realizable}. These results manifest the conceptual novelty of the tropical view on mechanism design and its utility, beyond supplying a complementary view on the network flow approach to mechanism design \cite{rochet1987necessary,vohra2011mechanism}. We remark that our theory applies equally to mechanism design with and without money, and is multidimensional from the outset. Importantly, our characterization of IC payments can be used to analyze mechanisms in the absence of revenue equivalence, \cf Example \ref{ex:information.rents}.

\bibliographystyle{plain}
\bibliography{mechanismDesign.bib}
\end{document}

%% file: TGMDfigure_firstexample.tex
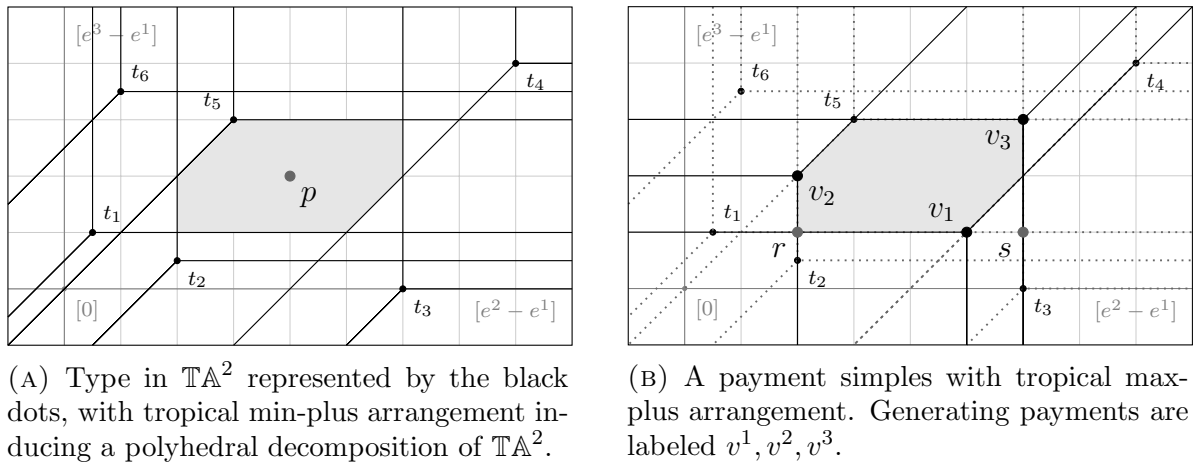
\begin{figure}
    \centering
\begin{subfigure}[t]{0.45\textwidth}
\centering
\begin{tikzpicture}[scale=0.75]
	\draw[help lines, step=1, black!20] (0,0) grid (10,6);
	
	\filldraw [black!10](3,2) -- (6,2) -- (7,3) -- (7,4) -- (4,4) -- (3,3);

\draw[black!50] (1,0) -- (1,6);
		\node[right, black!50] at (1,5.5) {\tiny $[e^3-e^1]$};
		\draw[black!50] (0,1) -- (10,1);
		\node[below, black!50] at (9,1) {\tiny $[e^2-e^1]$};
		\filldraw[black!50] (1,1) circle (1pt);
		\node[below right, black!50] at (1,1) {\tiny $[0]$};

	\filldraw (4,4) circle (1.5pt);
	\node[above left] at (4,4) {\tiny $t_5$};
	\draw [line width=0.15mm](4,6) -- (4,4) -- (0,0) -- (4,4) -- (10,4);

	\filldraw (2,4.5) circle (1.5pt);
	\node[above right] at (2,4.5) {\tiny $t_6$};
	\draw [line width=0.15mm](2,6) -- (2,4.5) -- (0,2.5) -- (2,4.5) -- (10,4.5);
	
	\filldraw (3,1.5) circle (1.5pt);
	\node[below right] at (3,1.5) {\tiny $t_2$};
	\draw [line width=0.15mm](3,6) -- (3,1.5) -- (1.5,0) -- (3,1.5) -- (10,1.5);
	
	\filldraw (1.5,2) circle (1.5pt);
	\node[above right] at (1.5,2) {\tiny $t_1$};
	\draw [line width=0.15mm](1.5,6) -- (1.5,2) -- (0,0.5) -- (1.5,2) -- (10,2);
	
	\filldraw (7,1) circle (1.5pt);
	\node[below right] at (7,1) {\tiny $t_3$};
	\draw [line width=0.15mm](7,6) -- (7,1) -- (6,0) -- (7,1) -- (10,1);
	
	\filldraw (9,5) circle (1.5pt);
	\node[below right] at (9,5) {\tiny $t_4$};
	\draw [line width=0.15mm](9,6) -- (9,5) --(10,5) --(9,5) (4,0) -- (9,5);
	
	\filldraw[black!60] (5,3) circle (2.5pt);
	\node[below right] at (5,3) {$p$};

	\draw[line width=0.1mm] (0,0) rectangle (10,6);
\end{tikzpicture}
\caption{Type in $\TA^2$ represented by the black dots, with tropical min-plus arrangement inducing a polyhedral decomposition of $\TA^2$. }
\end{subfigure}
    \hspace{0.5cm}
\begin{subfigure}[t]{0.45\textwidth}
\centering
\begin{tikzpicture}[scale=0.75]
	\draw[help lines, step=1, black!20] (0,0) grid (10,6);

	\draw[black!50] (1,0) -- (1,6);
		\node[right, black!50] at (1,5.5) {\tiny $[e^3-e^1]$};
		\draw[black!50] (0,1) -- (10,1);
		\node[below, black!50] at (9,1) {\tiny $[e^2-e^1]$};
		\filldraw[black!50] (1,1) circle (1pt);
		\node[below right, black!50] at (1,1) {\tiny $[0]$};

\filldraw [black!10](3,2) -- (6,2) -- (7,3) -- (7,4) -- (4,4) -- (3,3); 

	\filldraw (4,4) circle (1.5pt);
	\node[above left] at (4,4) {\tiny $t_5$};
	\draw [line width=0.25mm, dotted, black!60](4,6) -- (4,4)--(10,4);
	\draw [line width=0.25mm, dotted, black!60](4,4) -- (0,0);

	\filldraw (2,4.5) circle (1.5pt);
	\node[above right] at (2,4.5) {\tiny $t_6$};
	\draw [line width=0.25mm, dotted, black!60](2,6) -- (2,4.5) -- (10,4.5);
	\draw [line width=0.25mm, dotted, black!60](2,4.5) -- (0,2.5);
	
	\filldraw (3,1.5) circle (1.5pt);
	\node[below right] at (3,1.5) {\tiny $t_2$};
	\draw [line width=0.25mm, dotted, black!60](3,6) -- (3,1.5) -- (10,1.5);
	\draw [line width=0.25mm, dotted, black!60](3,1.5) -- (1.5,0);
	
	\filldraw (1.5,2) circle (1.5pt);
	\node[above right] at (1.5,2) {\tiny $t_1$};
	\draw [line width=0.25mm, dotted, black!60](1.5,6)  -- (1.5,2) -- (10,2);
	\draw [line width=0.25mm, dotted, black!60](1.5,2) -- (0,0.5);
	
	\filldraw (7,1) circle (1.5pt);
	\node[below right] at (7,1) {\tiny $t_3$};
	\draw [line width=0.25mm, dotted, black!60](7,6) -- (7,1) -- (10,1);
	\draw [line width=0.25mm, dotted, black!60](7,1) -- (6,0);
	
	\filldraw (9,5) circle (1.5pt);
	\node[below right] at (9,5) {\tiny $t_4$};
	\draw [line width=0.25mm, dotted, black!60](9,6) -- (9,5) -- (4,0) -- (9,5);
	\draw [line width=0.25mm, dotted, black!60](9,5) -- (10,5) ;

	\filldraw [black] (7,4) circle (2.5pt);
	\draw [line width=0.15mm](0,4) -- (7,4) -- (7,0) -- (7,4) -- (9,6);
	\node[below left] at (7,4) {$v_3$};
	\filldraw [black] (3,3) circle (2.5pt);
	\draw [line width=0.15mm](0,3) -- (3,3) -- (3,0) -- (3,3) -- (6,6);
	\node[below right] at (3,3) {$v_2$};
	\filldraw [black] (6,2) circle (2.5pt);
	\draw [line width=0.15mm](0,2) -- (6,2) -- (6,0) -- (6,2) -- (10,6);
	\node[above left] at (6,2) {$v_1$};

		\node[below left] at (3,2) {$r$};
	\filldraw [black!60] (3,2) circle (2.5pt);
			\node[below left] at (7,2) {$s$};
	\filldraw [black!60] (7,2) circle (2.5pt);

	\draw[line width=0.1mm] (0,0) rectangle (10,6);
\end{tikzpicture}
        \caption{A payment simples with tropical max-plus arrangement. Generating payments are labeled $v^1, v^2, v^3$.}
    \end{subfigure}
\caption{The type space $T$ consists of the points $\{t^1, \ldots, t^6 \}$ represented in $\TA^2$}\label{fig:motivation}
\end{figure}

%% file: TGMDfigure_re.tex
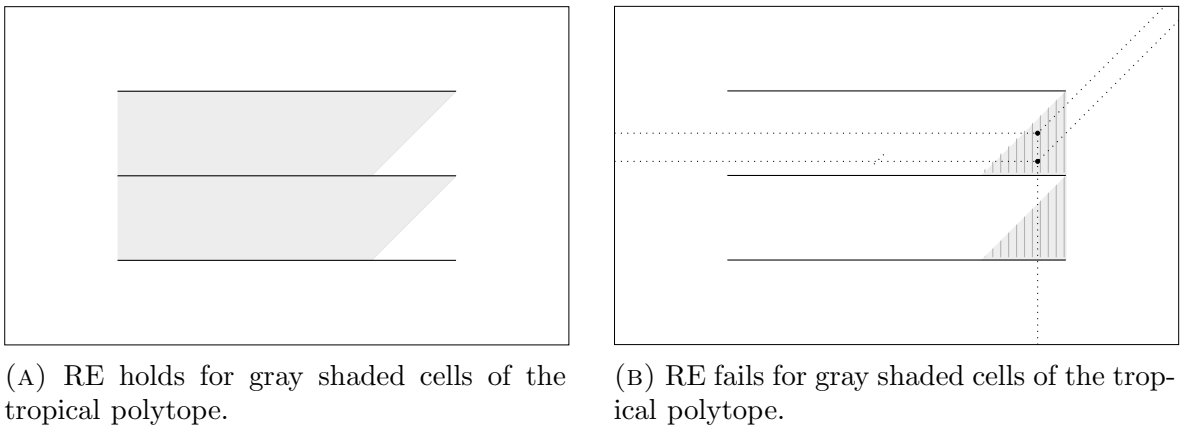
\begin{figure}
    \centering
    \begin{subfigure}[t]{0.45\textwidth}
        \centering
        \begin{tikzpicture}[scale=0.75]
		\filldraw [black!10, fill opacity=0.7] (2,1.5) -- (6.5,1.5) -- (8,3) -- (2, 3);	
		\filldraw [black!10, fill opacity=0.7] (2,3) -- (6.5,3) -- (8,4.5) -- (2, 4.5);	
		
		\draw[black] (2,1.5) -- (8,1.5);
		\draw[black] (2,3) -- (8,3);
		\draw[black] (2,4.5) -- (8,4.5);
		
		\draw[line width=0.1mm] (0,0) rectangle (10,6);
		
		\end{tikzpicture}
        \caption{RE holds for gray shaded cells of the tropical polytope.}
    \end{subfigure}%
    \hspace{0.5cm} 
    \begin{subfigure}[t]{0.45\textwidth}
        \centering
        \begin{tikzpicture}[scale=0.75]
		
		\filldraw [black!10, fill opacity=0.7] (6.5,1.5) -- (8,1.5) -- (8,3);
		\filldraw [black!10, fill opacity=0.7] (6.5,3) -- (8,3) -- (8,4.5);	
		\draw[black] (2,1.5) -- (8,1.5);
		\draw[black] (2,3) -- (8,3);
		\draw[black] (2,4.5) -- (8,4.5);
		
		\fill[pattern=vertical lines, pattern color=black!30] (6.5,3.05)-- (8,3.05)--(8,4.5)--(6.5,3.05);
	
		\fill[pattern=vertical lines, pattern color=black!30] (6.5,1.55)-- (8,1.55)--(8,3)--(6.5,1.5);
		
		\filldraw (7.5,3.75) circle (1pt);
		\draw [thin, dotted](7.5,0) -- (7.5,3.75) -- (9.75,6);
		\draw [thin, dotted](0,3.75) -- (7.5,3.75); 

		\filldraw (7.5,3.25) circle (1pt);
		\draw [thin, dotted](0,3.25) -- (7.5,3.25) -- (10,5.75);
		\draw [thin, dotted](4.6,3.2) -- (4.8,3.4) ;

		\draw[line width=0.1mm] (0,0) rectangle (10,6);
		
\end{tikzpicture}
        \caption{RE fails for gray shaded cells of the tropical polytope.}
    \end{subfigure}
\caption{The type space consists of three parallel lines. On this type space RE and non-RE outcome functions coexist.}\label{fig:re}
\end{figure}

%% file: TGMDfigure_hyperplanes.tex
\begin{example}[Tropical arrangements and sectors]\label{ex:arrangement}
Figure \ref{fig:tropical.hyperplanes}(\textsc{a}) depicts the arrangement $\minH(-t_0)$ and its sectors in $\TA^{2}$ identified with $\R^2$ via (\ref{eqn:tp.to.r}). The three lines are the three lower-dimensional sectors $\minH_{12}(-t_0)$, $\minH_{13}(-t_0)$, and $\minH_{23}(-t_0)$. The apex $t_0=\minH(-t)_{\{1,2,3 \}}$ is where these sectors meet. Panel (\textsc{b}) shows the max-plus arrangement $\maxH(-t_0)$ together with its sectors.
\end{example}

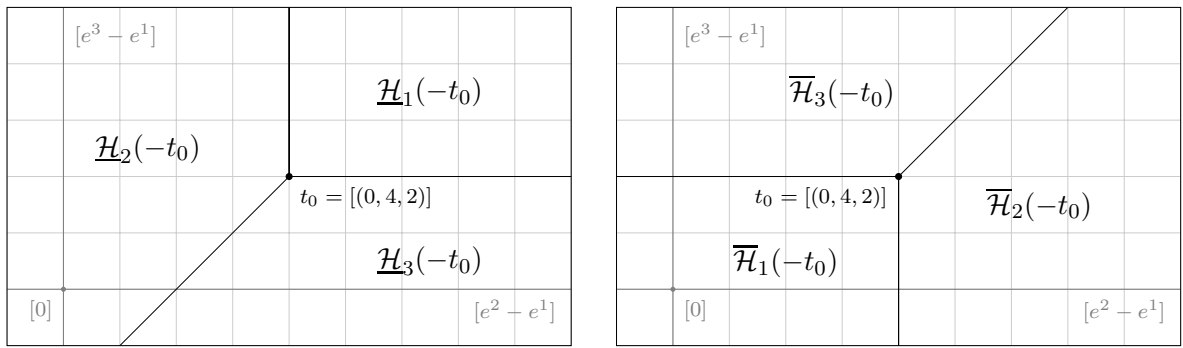
\begin{figure}[h]
    \centering

    \begin{subfigure}[t]{0.45\textwidth}
        \centering
        \begin{tikzpicture}[scale=0.75]
		\draw[help lines, step=1, black!20] (0,0) grid (10,6);
		\draw[black!50] (1,0) -- (1,6);
		\node[right, black!50] at (1,5.5) {\tiny $[e^3-e^1]$};
		\draw[black!50] (0,1) -- (10,1);
		\node[below, black!50] at (9,1) {\tiny $[e^2-e^1]$};
		\filldraw[black!50] (1,1) circle (1pt);
		\node[below left, black!50] at (1,1) {\tiny $[0]$};
		
		\filldraw (5,3) circle (1.5pt);
		\draw [line width=0.1mm](2,0) -- (5,3) -- (5,6) -- (5,3) -- (10,3);
		\node[below right] at (5,3) {\tiny $t_0= [(0, 4,2)]$};

		\node at (7.5,4.5) {\small $\minH_1(-t_0)$};
		\node at (7.5,1.5) {\small $\minH_3(-t_0)$};
		\node at (2.5,3.5) {\small $\minH_2(-t_0)$};
		\draw[line width=0.1mm] (0,0) rectangle (10,6);
\end{tikzpicture}
        \caption{A min-plus arrangement with its sectors.}
    \end{subfigure}%
    \hspace{0.5cm} 
    \begin{subfigure}[t]{0.45\textwidth}
        \centering
        \begin{tikzpicture}[scale=0.75]
		\draw[help lines, step=1, black!20] (0,0) grid (10,6);
		\draw[black!50] (1,0) -- (1,6);
		\node[right, black!50] at (1,5.5) {\tiny $[e^3-e^1]$};
		\draw[black!50] (0,1) -- (10,1);
		\node[below, black!50] at (9,1) {\tiny $[e^2-e^1]$};
		\filldraw[black!50] (1,1) circle (1pt);
		\node[below right, black!50] at (1,1) {\tiny $[0]$};

		\filldraw (5,3) circle (1.5pt);
		\draw [line width=0.1mm](8,6) -- (5,3) -- (5,0) -- (5,3) -- (0,3);
		\node[below left] at (5,3) {\tiny $t_0= [(0, 4,2)]$};	
		
		\node at (3,1.5) {\small $\maxH_1(-t_0)$};
		\node at (7.5,2.5) {\small $\maxH_2(-t_0)$};
		\node at (4,4.5) {\small $\maxH_3(-t_0)$};
		
		\draw[line width=0.1mm] (0,0) rectangle (10,6);
\end{tikzpicture}
        \caption{A max-plus arrangement with its sectors.}
    \end{subfigure}
\caption{Min-plus and max-plus arrangements with labeled sectors in two-dimensional tropical affine space $\TA^2$.}\label{fig:tropical.hyperplanes}
\end{figure}

%% file: TGMDfigure_arrangement.tex
\begin{figure}[b]
\centering
\begin{tikzpicture}[scale=0.75]
		\draw[help lines, black!20] (0,0) grid (21,9);
		
		\filldraw [black!12, fill opacity=0.7] (2.95,7.07) -- (18.05,7.07) -- (18.05,4.95) -- (14.05,0.95) -- (10.95,0.95)--(10.95,2.95)--(5.95,2.95)--(5.95,6.93)--(2.95,6.93);
		
		\draw[black!50] (1,0) -- (1,9);
		\node[right, black!50] at (1,8.5) {\tiny $[e^3-e^1]$};
		\draw[black!50] (0,1) -- (21,1);
		\node[below, black!50] at (20,1) {\tiny $[e^2-e^1]$};
		\filldraw[black!50] (1,1) circle (1pt);
		\node[below right, black!50] at (1,1) {\tiny $[0]$};

		\filldraw (3,7) circle (2.5pt);
		\draw [line width=0.1mm](0,4) -- (3,7) -- (3,9) -- (3,7) -- (21,7);
		\node[above right] at (3,7) {$t_2$};
		
		\filldraw (6,3) circle (2.5pt);
		\draw [line width=0.1mm](3,0) -- (6,3) -- (6,9) -- (6,3) -- (21,3);
		\node[above right] at (6,3) {$t_1$};
		
		\filldraw (11,1) circle (2.5pt);
		\draw [line width=0.1mm](10,0) -- (11,1) -- (11,9) -- (11,1) -- (21,1);
		\node[above left] at (11,1) {$t_3$};

		\filldraw (18,5) circle (2.5pt);
		\draw [line width=0.1mm](13,0) -- (18,5) -- (18,9) -- (18,5) -- (21,5);
		\node[above right] at (18,5) {$t_4$};
		
		\node at (8.5,5) {\tiny $\begin{pmatrix} 1&0&0&0\\0&0&1&1\\0&1&0&0\end{pmatrix}$};
		\node at (3.5,5) {\tiny $\begin{pmatrix} 0&1&0&0\\1&0&1&1\\0&1&0&0\end{pmatrix}$};
		\node at (14,5) {\tiny $\begin{pmatrix} 1&0&1&0\\0&0&0&1\\0&1&0&0\end{pmatrix}$};
		\node at (12.9,2) {\tiny $\begin{pmatrix} 0&0&1&0\\0&0&0&1\\1&1&0&0\end{pmatrix}$};
		\node at (8,1.5) {\tiny $\begin{pmatrix} 1&0&1&0\\0&0&1&1\\1&1&0&0\end{pmatrix}$};

	\draw [line width=0.2mm, dotted](5.5,0) -- (5.5,7) -- (7.5,9);
	\draw [line width=0.2mm, dotted](0,7) -- (5.5,7); 
	
	\draw [line width=0.2mm, dotted](16,0) -- (16,4) -- (21,9);
		\draw [line width=0.2mm, dotted](0,4) -- (16,4);
		
		\draw [line width=0.2mm, dotted](11,0) -- (11,2.5) -- (17.5,9);
		\draw [line width=0.2mm, dotted](0,2.5) -- (11,2.5);

		\draw
    [->,thin, dashed](5,5) to [out=0, in=225] (5.3,6.8);

		\draw
    [->,thin, dashed](9.5,1.5) to [out=0, in=225] (10.9,2.9);

		\filldraw (16,4) circle (1pt);
		\node[below left] at (16,4) { $p$};
		\filldraw (5.5,7) circle (1pt);
		\node[above] at (5.5,7)  {$r$};
		\filldraw (6,7) circle (1pt);
		\node[above right] at (6,7)  {$u$};
		\filldraw (11,3) circle (1pt);
		\node[above right] at (11,3) { $s$};
		\filldraw (11,2.5) circle (1pt);
		\node[right] at (11,2.5) { $q$};
	
		\draw[line width=0.1mm] (0,0) rectangle (21,9);
\end{tikzpicture}
\caption{A min-plus arrangement on four points labeled  with some of its covectors. The tropical polytope generated by the points is shaded gray.}\label{fig:hyperplane.arrangement}
\end{figure}
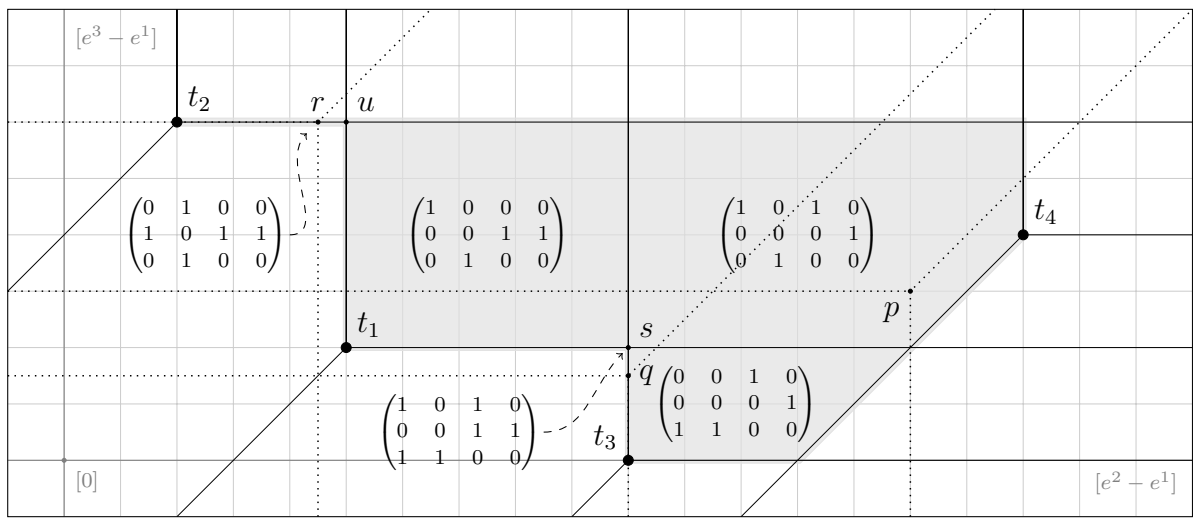

\begin{example}[Covectors, Tropical Polytopes]\label{ex:covectors} 
	
Figure \ref{fig:hyperplane.arrangement} depicts type set $T=\{t_1, t_2, t_3, t_4\}$ together with the min-plus arrangement $\minH(-T)$. The subcomplex of bounded cells is shaded gray. It is the tropical polytope generated by $T$, $\tconv(T)$. Matrices in the figure are the covectors of a selection of cells in $\tconv(T)$. Consider the matrices
\begin{equation*}
	h=\begin{pmatrix}
0&0&0&0\\ 0&0&1&1\\1&1&0&0	
	\end{pmatrix},\quad 
	h'=\begin{pmatrix}
0&0&1&0\\ 0&0&1&1\\1&1&0&0	
	\end{pmatrix}.
	\end{equation*} 
The matrix $h$ defines a cell $\tau$ of $\tconv(T)$ by restricting $\minH_3(-t_1)\cap \minH_3(-t_2)\cap\minH_2(-t_3)\cap\minH_2(-t_4)$ to $\tconv(T)$, provided the intersection is not empty. In fact, $\tau$ is the line connecting the points labeled $s$ and $t_3$. But the covector of this cell is $h'$ which is different from $h$. Thus $h$ is not a covector of $\tconv(T)$.
\end{example}

%% file: TGMDfigure_construction.tex
\begin{example}[Geometric and algebraic construction of IC mechanisms]\label{ex:construction1} This example illustrates Propositions \ref{prop:covector} and \ref{prop:projection.payments}. In Figure \ref{fig:hyperplane.arrangement}, let $\sigma$ be the cell of $\tconv(T)$ that contains $p$. Since $\sigma$ is full-dimensional, by Proposition \ref{prop:covec.full.dim}, there is a unique IC $g: T \to [3]$ with this payment $p$, namely, $$g(t_1)=g(t_3)=1, g(t_4)=2, g(t_2)=3.$$ The max-plus arrangement $\maxH(-p)$ at $p$ (dotted) can be used to define $g$ geometrically: $g(t)~=~i$ if and only if $t\in T\cap \maxH_i(-p)$. Since $g$ is onto, by Proposition \ref{prop:projection.payments}, $P(g)=\sigma$. 
\end{example}

\begin{example}[Construction of IC mechanisms (cont.)]\label{ex:construction2}
In Figure \ref{fig:hyperplane.arrangement}, consider the point $r$. Let $\nu$ be the cell of $\tconv(T)$ that contains $r$ in its relative interior. This cell is the line segment between $t_2$ and $u$. The graph that corresponds to the covector of $\nu$ has two subgraphs defining outcome functions on $T$, whose adjacency matrices are
\begin{equation*}
	f_1=\begin{pmatrix}
0&0&0&0\\1&0&1&1\\0&1&0&0
	\end{pmatrix},\quad 
	f_2=\begin{pmatrix}
0&1&0&0\\1&0&1&1\\0&0&0&0
	\end{pmatrix}.
	\end{equation*} 
By Proposition \ref{prop:covector}, $\mathcal G(r)$ consists of two IC outcome functions $f_1, f_2: T \to [3]$, with
$$f_1(t_2)=1, f_1(t_1)=f_1(t_3)=f_1(t_4)=2, \quad \mbox{ and } \quad f_2(t_2)=3, f_2(t_1)=f_2(t_3)=f_2(t_4)=2.$$ 
Again, $\maxH(-r)$ can be used to determine $\mathcal{G}(r)$. Here, $t_2\in T\cap \maxH_1(-r)\cap\maxH_3(-r)$, reflecting the fact that $t_2$ can be assigned to either outcome 1 or 3 without violating IC. Then, modulo addition of constants, $\mathcal P(f_1)\subset \R^{\{1,2\}}$, and $\mathcal P(f_2)\subset \R^{\{2,3\}}$ are one-dimensional polytopes, which can be lifted to $\Plift(f_1)=\Plift(f_1)=\nu$ in $\R^3$ by Proposition \ref{prop:projection.payments}. The lifting preserves the dimension of the cells. 
\end{example}

%% file: TGMDfigure_onto.tex
\begin{figure}[b]
\centering
\begin{tikzpicture}[scale=0.75]
		
		\filldraw [black!12, fill opacity=0.7] (2.95,7.07) -- (18.05,7.07) -- (18.05,4.95) -- (14.05,0.95) -- (10.95,0.95)--(10.95,2.95)--(5.95,2.95)--(5.95,6.93)--(2.95,6.93);	
		
		\filldraw (3,7) circle (2.5pt);
		\draw [line width=0.1mm](0,4) -- (3,7) -- (3,9) -- (3,7) -- (21,7);
		\node[above right] at (3,7) {$t_2$};
		
		\filldraw (6,3) circle (2.5pt);
		\draw [line width=0.1mm](3,0) -- (6,3) -- (6,9) -- (6,3) -- (21,3);
		\node[above right] at (6,3) {$t_1$};
		
		\filldraw (11,1) circle (2.5pt);
		\draw [line width=0.1mm](10,0) -- (11,1) -- (11,9) -- (11,1) -- (21,1);
		\node[above left] at (11,1) {$t_3$};

		\filldraw (18,5) circle (2.5pt);
		\draw [line width=0.1mm](13,0) -- (18,5) -- (18,9) -- (18,5) -- (21,5);
		\node[above right] at (18,5) {$t_4$};
		
		\node at (8.5,5) {$g_1$};
		\node at (14,5) {$g_2$};
		\node at (13,2) {$g_3$};
		
		\node at (4.5,8) {$h_1$};
		\node at (8.5,8) {$h_2$};
		\node at (14,8) {$h_3$};
		\node at (19.5,8) {$h_4$};
		\node at (19.5,6) {$h_5$};
		\node at (19.5,4) {$h_6$};
		\node at (18.3,2) {$h_7$};
		\node at (17.5,0.5) {$h_8$};
		\node at (12,0.5) {$h_9$};
		\node at (7.5,2) {$h_{10}$};
		\node at (4,5) {$h_{11}$};
		\node at (1.5,7.5) {$h_{12}$};
		\draw[line width=0.1mm] (0,0) rectangle (21,9);
\end{tikzpicture}
\caption{A min-plus arrangement on four points. The tropical polytope generated by the points is shaded gray. All full-dimensional cells of the arrangement are labeled by their covectors.}\label{fig:onto.mechanisms}
\end{figure}
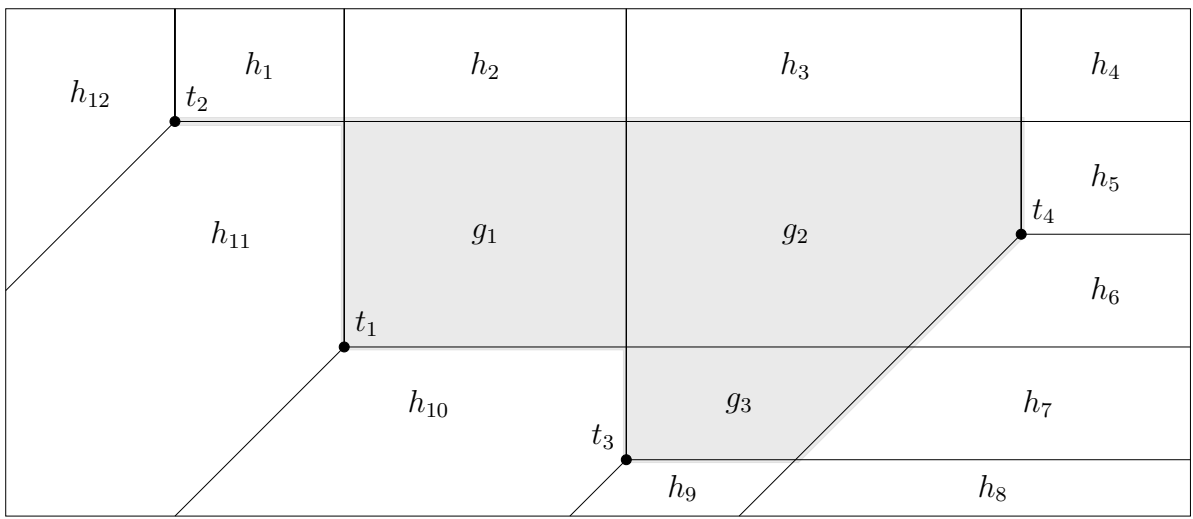

\begin{example}\label{ex:onto.mechanisms}
Consider Figure \ref{fig:onto.mechanisms}, which reproduces parts of Example \ref{ex:covectors}. Again, the tropical polytope on $T=\{t_1, t_2, t_3, t_4 \}$ is shaded grey. The type space is generic, since no three points lie on the same hyperplane, and no two points project onto the same point on either of the axes. There are three bounded, full-dimensional cells, and thus exactly three onto mechanism by Theorem \ref{thm:main.generic},
\begin{equation*}
	g_1=\begin{pmatrix}
1&0&0&0\\0&0&1&1\\0&1&0&0
	\end{pmatrix},\quad 
	g_2=\begin{pmatrix}
1&0&1&0\\0&0&0&1\\0&1&0&0
	\end{pmatrix},\quad 
	g_3=\begin{pmatrix}
0&0&1&0\\1&1&0&0\\0&0&0&1
	\end{pmatrix}.
	\end{equation*} 
Moreover, there are 12 non-surjective mechanisms, each associated with one full-dimensional cell of $\minH(-T)$. 
\begin{align*}
	&h_1=\begin{pmatrix}
1&0&0&0\\ 0&1&1&1\\0&0&0&0	
	\end{pmatrix},\quad 
	&h_2=\begin{pmatrix}
1&1&0&0\\ 0&0&1&1\\0&0&0&0	
	\end{pmatrix}, \quad 
	&h_3=\begin{pmatrix}
1&1&1&0\\ 0&0&0&1\\0&0&0&0	
	\end{pmatrix},\quad 
	&h_4=\begin{pmatrix}
1&1&1&1\\ 0&0&0&0\\0&0&0&0	
	\end{pmatrix}\\
	&h_5=\begin{pmatrix}
1&1&0&1\\ 0&0&0&0\\0&0&1&0	
	\end{pmatrix},\quad 
	&h_6=\begin{pmatrix}
1&0&1&0\\ 0&0&0&0\\0&1&0&1	
	\end{pmatrix}, \quad 
	&h_7=\begin{pmatrix}
0&0&1&0\\ 0&0&0&0\\1&1&0&1	
	\end{pmatrix},\quad 
	&h_8=\begin{pmatrix}
0&0&0&0\\ 0&0&0&0\\1&1&1&1	
	\end{pmatrix}\\
	&h_9=\begin{pmatrix}
0&0&0&0\\ 0&0&0&1\\1&1&1&0	
	\end{pmatrix},\quad 
	&h_{10}=\begin{pmatrix}
0&0&0&0\\ 0&0&1&1\\1&1&0&0	
	\end{pmatrix}, \quad 
	&h_{11}=\begin{pmatrix}
0&0&0&0\\ 1&1&0&1\\0&0&1&0	
	\end{pmatrix},\quad 
	&h_{12}=\begin{pmatrix}
0&0&0&0\\ 1&1&1&1\\0&0&0&0	
	\end{pmatrix}
	\end{align*} 
	Note that for these mechanisms, the set of IC payments can be identified with the intersection of the corresponding cell in $\minH(-T)$ with the tropical polytope $\tconv(T)$. Thus, as Corollary \ref{cor:cardinality} predicts, there are 15 mechanisms on $T$ in total, and 3 onto mechanisms. 
\end{example}

%% file: TGMDfigure_payments.tex
\begin{example}[Generating Payments]\label{ex:generating.payments}
	Consider the outcome function $g$ from Example \ref{ex:construction1}. Its IC payment simplex $\mathcal P(g)$ is the cell containing $p$ in its interior in Figure \ref{fig:hyperplane.arrangement}. The generators can be calculated using (\ref{eqn:geometric.kleene.star}) and are given by $v_1=[(0, 17, 6)]$, $v_2=[(0, 10, 6)]=[(-10, 0, -4)]$, $v_3=[(0, 15,2)]=[(-2, 13, 0)]$. If $p'=(0,17,6)\in v_1$ is the IC price for the mechanism $(g,p)$, then $p'_2=17$ is the maximal price that supports outcome 2. If it were increased, then $t_4$ would deviate and announce $t_1$ or $t_3$. By Theorem \ref{thm:tropicaleigenspacegenerator}, the payment simplex $\mathcal P(g)$ equals the \emph{min-plus} linear span of the generators $v_1, v_2, v_3$ or equivalently the union of bounded cells in the \emph{max-plus} arrangement on these points, see Remark \ref{rem:duality}. Both characterize the min-plus convex hull $\underline{\mathsf{tconv}}(v_1, v_2, v_3)$ by \cite[Thm. 15]{develin2004tropical}. For instance, the IC price $p=[(0, 15, 3)]$ depicted in the figure can be expressed as a min-plus linear combination as $p=5\odot v_1 \underline \oplus 0\odot v_2\underline\oplus 1\odot v_3$. However, the classical convex hull of the generating payments is a strict subset of $\mathcal{P}(g)$. 
\end{example}

%% file: TGMDfigure_dimensiongraph.tex
\begin{example}[Dimension graph and number of generating payments] \label{ex:dimension} 
We continue with the setting of Example \ref{ex:generating.payments}.
The dimension graph of $p$ is given by 3 nodes with no edges, since $T\cap \maxH_{ij}(-p)=\emptyset$ for all $i\neq j$, see Figure \ref{fig:hyperplane.arrangement} . By Theorem \ref{thm:dimension} we conclude that there are exactly three generating payments for $g$, which were calculated in Example \ref{ex:generating.payments}. Recall that the cell $\mathcal P(g)$ is depicted in $\TA^2$, where its dimension is two. Accounting for the quotiented lineality space $(0, 0, 0)\odot \R$, the cell has dimension three in $\R^3$. So this number agrees with the dimension of $\mathcal P (g)$ as a polyhedron in $\R^3$ as expected from Theorem \ref{thm:tropicaleigenspacegenerator}.
\end{example}

%% file: TGMDfigure_dimensiongraph2.tex
\begin{example}[Example \ref{ex:dimension} continued] \label{ex:dimension2} 
Consider now the mechanism $f_1$ from Example~\ref{ex:construction2} and the cell $\nu=\Plift(f_1)$ given by the line segment connecting $t_2$ and $u$. Inspecting Figure \ref{fig:hyperplane.arrangement}, we note  that the dimension graph of any point $r$ in the relative interior of $\nu$ has edges $\{1,3\}$, hence two connected components. The set of IC payments is generated by exactly two prices which can be seen to be the points $t_2$ and $u$. Thus, calculating the columns of $(L^{f_1})^*\in \R^{3\times 3}$ from $\Plift(f_1)$ by means of  (\ref{eqn:geometric.eigenspace}), returns two tropically linearly dependent columns, reflecting the fact that the cell is not of full dimension.
\end{example}

%% file: TGMDfigure_re2.tex
\begin{figure}
    \centering
    \begin{subfigure}[t]{0.45\textwidth}
        \centering
        \begin{tikzpicture}[scale=0.75]
		
		\filldraw [black!10, fill opacity=0.7] (2,1.5) -- (8,1.5) -- (8,4.5) -- (2, 4.5);	
		
		\draw[black!20] (2,1.5) -- (8,1.5);
		\draw[black!20] (2,3) -- (8,3);
		\draw[black!20] (2,4.5) -- (8,4.5);

		\draw[black!20] (2,1.5) -- (2,4.5);
		\draw[black!20] (3,1.5) -- (3,4.5);
		\draw[black!20] (4,1.5) -- (4,4.5);
		\draw[black!20] (5,1.5) -- (5,4.5);
		\draw[black!20] (6,1.5) -- (6,4.5);
		\draw[black!20] (7,1.5) -- (7,4.5);
		\draw[black!20] (8,1.5) -- (8,4.5);
		
		\draw[black!20] (8,4.5) -- (5,1.5);
		\draw[black!20] (7,4.5) -- (4,1.5);
		\draw[black!20] (6,4.5) -- (3,1.5);
		\draw[black!20] (5,4.5) -- (2,1.5);
		\draw[black!20] (4,4.5) -- (2,2.5);
		\draw[black!20] (3,4.5) -- (2,3.5);
		
		\draw[black!20] (8,3) -- (6.5,1.5);
		\draw[black!20] (7,3) -- (5.5,1.5);
		\draw[black!20] (6,3) -- (4.5,1.5);
		\draw[black!20] (5,3) -- (3.5,1.5);
		\draw[black!20] (4,3) -- (2.5,1.5);
		\draw[black!20] (3,3) -- (2,2);
		
		\filldraw (2,1.5) circle (1.5pt);
		\filldraw (3,1.5) circle (1.5pt);
		\filldraw (4,1.5) circle (1.5pt);
		\filldraw (5,1.5) circle (1.5pt);
		\filldraw (6,1.5) circle (1.5pt);
		\filldraw (7,1.5) circle (1.5pt);
		\filldraw (8,1.5) circle (1.5pt);

		\filldraw (2,3) circle (1.5pt);
		\filldraw (3,3) circle (1.5pt);
		\filldraw (4,3) circle (1.5pt);
		\filldraw (5,3) circle (1.5pt);
		\filldraw (6,3) circle (1.5pt);
		\filldraw (7,3) circle (1.5pt);
		\filldraw (8,3) circle (1.5pt);

		\filldraw (2,4.5) circle (1.5pt);
		\filldraw (3,4.5) circle (1.5pt);
		\filldraw (4,4.5) circle (1.5pt);
		\filldraw (5,4.5) circle (1.5pt);
		\filldraw (6,4.5) circle (1.5pt);
		\filldraw (7,4.5) circle (1.5pt);
		\filldraw (8,4.5) circle (1.5pt);
						
		\draw[line width=0.1mm] (0,0) rectangle (10,6);
\end{tikzpicture}
        \caption{Discrete approximation of type space with tropical polytope and cells.}
    \end{subfigure}%
    \hspace{0.5cm} 
    \begin{subfigure}[t]{0.45\textwidth}
        \centering
        \begin{tikzpicture}[scale=0.75]
		
		\filldraw [black!10, fill opacity=0.7] (2,1.5) -- (8,1.5) -- (8,4.5) -- (2, 4.5);	
		\draw[black] (2,1.5) -- (8,1.5);
		\draw[black] (2,3) -- (8,3);
		\draw[black] (2,4.5) -- (8,4.5);
		
		\fill[pattern=vertical lines, pattern color=black!30] (6.5,3.05)-- (8,3.05)--(8,4.5)--(6.5,3.05);
	
		\fill[pattern=vertical lines, pattern color=black!30] (6.5,1.55)-- (8,1.55)--(8,3)--(6.5,1.5);
		
		\filldraw (7.5,3.75) circle (1pt);
		\draw [thin, dotted](7.5,0) -- (7.5,3.75) -- (9.75,6);
		\draw [thin, dotted](0,3.75) -- (7.5,3.75); 

		\filldraw (7.5,3.25) circle (1pt);
		\draw [thin, dotted](0,3.25) -- (7.5,3.25) -- (10,5.75);
		\draw [thin, dotted](4.6,3.2) -- (4.8,3.4) ;

	\node [right] at (8,3) {\tiny RE fails}; 
	\draw
    [->,thin, dashed](8.75,2.75) to [out=270, in=0] (8.1,2.5);
    \draw
    [->,thin, dashed](8.75,3.25) to [out=90, in=0] (8.1,3.5);
    
    \node at (2,1) {\tiny RE holds}; 
	\draw
    [->,thin, dashed](3,1) to [out=0, in=270] (3.5,2);

	\draw[thin, dashed](6.5,1.55)-- (8,1.55)--(8,3)--(6.5,1.5) ;
	\draw[thin, dashed](6.5,3.05)-- (8,3.05)--(8,4.5)--(6.5,3.05) ;

		\draw[line width=0.1mm] (0,0) rectangle (10,6);
		
\end{tikzpicture}
        \caption{Type space with tropical polytope.}
    \end{subfigure}
\caption{A type space on which RE and non-RE outcome functions coexist.}\label{fig:re.2}
\end{figure}
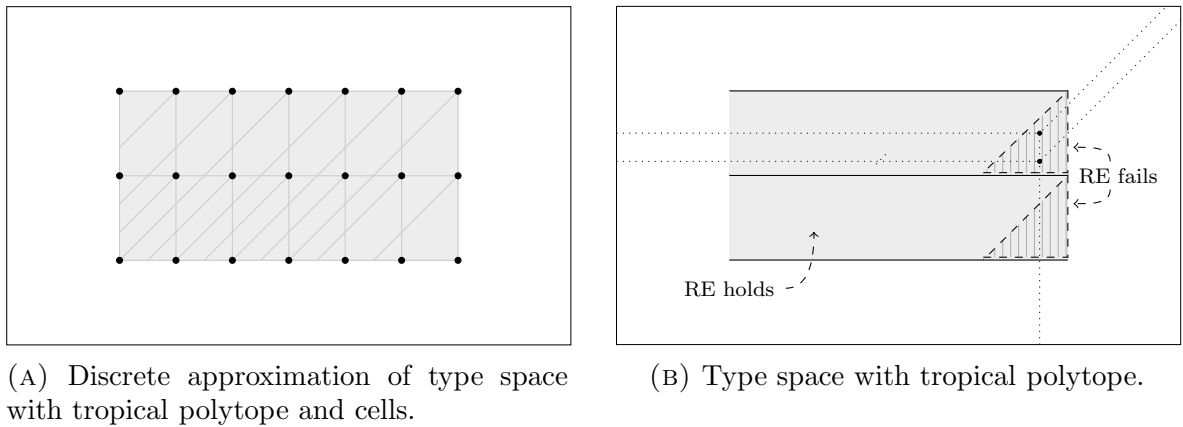

\begin{example}[RE via Corollary \ref{cor:re}] \label{ex:re}
Consider Figure \ref{fig:re.2} where the type space is given by three parallel copies of an interval, depicted using black lines in panel (\textsc b). The black dots in panel (\textsc a) are a discrete approximation, for which we also indicated the cell structure of the tropical polytope. The tropical polytopes are shaded gray in both panels. A limiting argument using finite approximations shows that all cells converge to points in the region which is not shaded in the line pattern, see Example \ref{ex:limit}. For these cells RE holds by Corollary \ref{cor:re}. In the two triangular regions shared with the line pattern, RE fails. There cells converge to vertical line segments. Using Definition \ref{defn:dimension} one verifies that the dimension graph has edges between the nodes labeled 1 and 2, while the node with label 3 is isolated. Hence the payment simplex has dimension two. Indeed,  in panel (\textsc b) we depicted two prices that support the same outcome function yet differ modulo addition of constants.  
\end{example}

%% file: TGMDfigure_informationrents.tex
\begin{example}[Equilibrium utility and extremal payments]\label{ex:information.rents.2} 
Consider the outcome function $g$ from Example \ref{ex:construction1}. Following Remark \ref{rem:affine.space} we normalize with respect to outcome 1, thus types and prices are depicted as cosets $t=[(0, t_2-t_1, t_3-t_1)]$, $p=[(0, p_2-p_1, p_3-p_1)]$. If we choose a payment coset $p=[(0, p_2-p_1, p_3-p_1)] \in \Plift(g)$ and $\tilde p_1\in \R$, we can fix the representative $\tilde p\in p$ with first coordinate $\tilde p_1$ to be the mechanism's payment function. Then the equilibrium utility function is given by $$u([g(t),p(t)],t)=(t-\tilde p)_{g(t)}= (t_j-t_1)-(\tilde p_j-\tilde p_1)+t_1-\tilde p_1 \quad \mbox{for}~ t\in g^{-1}(j).$$ Geometrically this is the difference of the projections of $t-p$ onto the $j$'th axis, normalized by $t_1-\tilde p_1$. The extremal payments of \cite{kos2013extremal} for $g$ with ``anchor type" $\tilde t\in g^{-1}(1)$ are calculated geometrically as follows. For each type in $g^{-1}(j)$, the lower bound is the maximal $j$'th axis-aligned difference to the vectors in the payment simplex $\Plift(g)$, with $\tilde p_1 = \tilde t_1$. Similarly the upper bound is calculated as the axis-aligned minimal difference to vectors in the simplex. It is easy to see that the lower and upper bound are the representatives of the vectors in the faces of $\Plift(g)$ supported by $[(0,-1, -1)]$ and $[(0,1,1)]$, respectively, with $\tilde p_1=\tilde t_1$ as first coordinate. These payments bound any IC payment rule.\end{example}

%% file: TGMDfigure_informationrents2.tex
\begin{example}[Equilibrium utility for non-RE mechanisms]\label{ex:information.rents} 

The purpose of this example is to illustrate how the geometry of cells determines the freedom of mechanisms to choose equilibrium utilities in the absence of RE. Note that the geometry of payment simplices in $\tconv(T)$ is very restricted, only cells with facet normals $e_i-e_j$, for $i\neq j$ and $\pm e_i$ in $\R^{m-1}$ can arise, \cite{joswig2010tropical, tran2017enumerating}. Via (\ref{eqn:pz.min}) each point in a payment simplex corresponds to some feasible utility function under the outcome function. In economics terms, the subset of directions occurring in the relative interior of a cell determine how the mechanism can trade off equilibrium utilities between types. 
Consider for instance the outcome function $g$ from Example \ref{ex:construction1}. Any payment $p=[(0, p_2-p_1, p_3-p_1)]$ in the relative interior of $\Plift(g)$ can be perturbed in the direction of $[(0,1,0)]$ and $[(0,0,1)]$ within the cell so that the mechanism can perturb the utility of types in $g^{-1}(2)$ and $g^{-1}(3)$ independently of each other. Moreover the normalizing coordinate $p_1$ can be  set freely without changing the coset of $p$, so that also the utility of types in $g^{-1}(1)$ can be chosen. 
In contrast lower-dimensional payment sets limit the ability of the mechanism to distribute utility. Consider Figure \ref{fig:re.2} (\textsc b). As noted in Example~\ref{ex:re}, the cells in the two dashed triangular regions are lines in the $[(0,0,1)]$ direction. We depicted two max-plus hyperplanes (dotted) that can be used to define the same IC surjective outcome function. The mechanism designer can extract utility of the types receiving outcome 3 by perturbing prices in the $[(0,0,1)]$ direction within this cell, but cannot extract utility from types receiving outcome 2 without changing the outcome function. 

While the geometry of payment simplices is determined by the geometry of the type space as a whole, each simplex is completely specified by its generating payments. Hence, the generating payments are sufficient to study all feasible equilibrium utility functions under a given outcome function. 
For recent works on mechanism design without revenue equivalence we refer to \cite{kos2013extremal, carbajal2013mechanism}.  We also point to the works of \cite{myerson1981optimal, rahman2014surplus, chen2013genericity, heifetz2006generic, mcafee1992correlated} for detailed studies of rent and surplus extraction using other techniques. \end{example}

%% file: TGMDfigure_allocationmatrix.tex
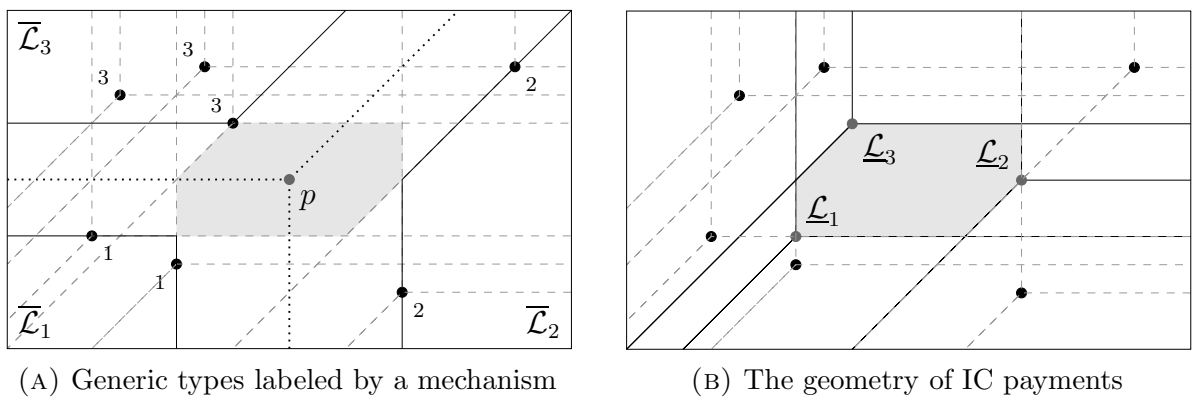
\begin{figure}[h!]
    \centering
\begin{subfigure}[t]{0.45\textwidth}
\centering
\begin{tikzpicture}[scale=0.75]
	
	\filldraw [black!10](3,2) -- (6,2) -- (7,3) -- (7,4) -- (4,4) -- (3,3); 
	
	\filldraw (4,4) circle (2.5pt);
	\node[above left] at (4,4) {\tiny $3$};
	\draw [line width=0.1mm, black!40, dashed](4,6) -- (4,4) -- (0,0) -- (4,4) -- (10,4);
	
	\filldraw (3.5,5) circle (2.5pt);
	\node[above left] at (3.5,5) {\tiny $3$};
	\draw [line width=0.1mm, black!40, dashed](3.5,6) -- (3.5,5) -- (0,1.5) -- (3.5,5) -- (10,5);
	
	\filldraw (2,4.5) circle (2.5pt);
	\node[above left] at (2,4.5) {\tiny $3$};
	\draw [line width=0.1mm, black!40, dashed](2,6) -- (2,4.5) -- (0,2.5) -- (2,4.5) -- (10,4.5);
	
	\filldraw (3,1.5) circle (2.5pt);
	\node[below left] at (3,1.5) {\tiny $1$};
	\draw [line width=0.1mm, black!40, dashed](3,6) -- (3,1.5) -- (1.5,0) -- (3,1.5) -- (10,1.5);
	
	\filldraw (1.5,2) circle (2.5pt);
	\node[below right] at (1.5,2) {\tiny $1$};
	\draw [line width=0.1mm, black!40, dashed](1.5,6) -- (1.5,2) -- (0,0.5) -- (1.5,2) -- (10,2);
	
	\filldraw (7,1) circle (2.5pt);
	\node[below right] at (7,1) {\tiny $2$};
	\draw [line width=0.1mm, black!40, dashed](7,6) -- (7,1) -- (6,0) -- (7,1) -- (10,1);
	
	\filldraw (9,5) circle (2.5pt);
	\node[below right] at (9,5) {\tiny $2$};
	\draw [line width=0.1mm, black!40, dashed](9,6) -- (9,5) -- (4,0) -- (9,5);
	
	\draw [line width=0.25mm, dotted](5,0) -- (5,3) -- (8,6);
	\draw [line width=0.25mm, dotted](0,3) -- (5,3); 
	\filldraw[black!60] (5,3) circle (2.5pt);
	\node[below right] at (5,3) {$p$};

	\draw [line width=0.1mm, black](0,4) -- (4,4) -- (6,6);
	\draw [line width=0.1mm, black](0,2) -- (3,2) -- (3,0);
	\draw [line width=0.1mm, black](7,0) -- (7,3) -- (10,6);
	
	\node[above right] at (0,0) {$\maxL_1$};
	\node[above left] at (10,0) {$\maxL_2$};
	\node[below right] at (0,6) {$\maxL_3$};
	
	\draw[line width=0.1mm] (0,0) rectangle (10,6);
\end{tikzpicture}
\caption{Generic types labeled by a mechanism}\label{fig:separate}
\end{subfigure}
    \hspace{0.5cm}
\begin{subfigure}[t]{0.45\textwidth}
\centering
\begin{tikzpicture}[scale=0.75]
	
	\filldraw [black!10](3,2) -- (6,2) -- (7,3) -- (7,4) -- (4,4) -- (3,3); 
	
	\node[above right] at (3,2) {$\minL_1$};
	\draw [line width=0.1mm](3,6) -- (3,2) -- (1,0) -- (3,2) -- (10,2);
	\filldraw [black!60] (3,2) circle (2.5pt);	
	
	\node[above left] at (7,3) {$\minL_2$};
	\draw [line width=0.1mm](7,6) -- (7,3) -- (4,0) -- (7,3) -- (10,3);
	\filldraw [black!60] (7,3) circle (2.5pt);
	
	\node[below right] at (4,4) {$\minL_3$};
	\draw [line width=0.1mm](4,6) -- (4,4) -- (0,0) -- (4,4) -- (10,4);
	\filldraw [black!60] (4,4) circle (2.5pt);

	\filldraw (3.5,5) circle (2.5pt);
	\draw [line width=0.1mm, black!40, dashed](3.5,6) -- (3.5,5) -- (0,1.5) -- (3.5,5) -- (10,5);
	
	\filldraw (2,4.5) circle (2.5pt);
	\draw [line width=0.1mm, black!40, dashed](2,6) -- (2,4.5) -- (0,2.5) -- (2,4.5) -- (10,4.5);
	
	\filldraw (3,1.5) circle (2.5pt);
	\draw [line width=0.1mm, black!40, dashed](3,6) -- (3,1.5) -- (1.5,0) -- (3,1.5) -- (10,1.5);
	
	\filldraw (1.5,2) circle (2.5pt);
	\draw [line width=0.1mm, black!40, dashed](1.5,6) -- (1.5,2) -- (0,0.5) -- (1.5,2) -- (10,2);
	
	\filldraw (7,1) circle (2.5pt);
	\draw [line width=0.1mm, black!40, dashed](7,6) -- (7,1) -- (6,0) -- (7,1) -- (10,1);
	
	\filldraw (9,5) circle (2.5pt);
	\draw [line width=0.1mm, black!40, dashed](9,6) -- (9,5) -- (4,0) -- (9,5);

	\draw[line width=0.1mm] (0,0) rectangle (10,6);
\end{tikzpicture}
        \caption{The geometry of IC payments}\label{fig:separate}
    \end{subfigure}
\caption{Geometric construction of allocation matrices. The type set consists of the black dots that are apices of the min-plus hyperplanes drawn using dashed lines. In Panel \textsc{(a)}, types have been labeled according to their outcome under an outcome function $g$. In Panel \textsc{(b)}, the gray dots are the rows of the matrix $L^g$, viewed as points in $\TA^2$. The set of IC payments of $g$ is shaded gray in both panels. Figure accompanies Example \ref{ex:separate}.}\label{fig:separate}
\end{figure}

\begin{example}[Geometric construction of allocation matrices] \label{ex:separate}
	Figure \ref{fig:separate} \textsc{(a)} shows how to construct the apices of the sectors $\maxL_i$ geometrically. The set $T$ consists of the generic black points. The dashed lines are hyperplanes of $\minH(-T)$. The labels next to the points define an outcome function $g: T \to [3]$. For $i \in [3]$, the heavy black lines define the boundary of the max-plus sectors $\overline{\mathcal{L}}_i$, whose apex is the $i$-th row of $L^g$. For example, for $i=1$, the first row of $L^g$ can be written as $ L^g_1 = (0, ~ \sup_{t\in g^{-1}(1)}\{t_2-t_1 \},~ \sup_{t\in g^{-1}(1)}\{t_3-t_1 \})$. The set $\mathcal{P}(g)$ equals the cell shaded gray, computed using Proposition  \ref{prop:projection.payments} and Definition \ref{def:p.lift}. Panel\textsc{(b)} verifies that this cell equals $\bigcap_{i=1}^3\minL_i$, as stipulated by Corollary \ref{cor:L.payment}. 
\end{example}

%% file: TGMDfigure_realize.tex
	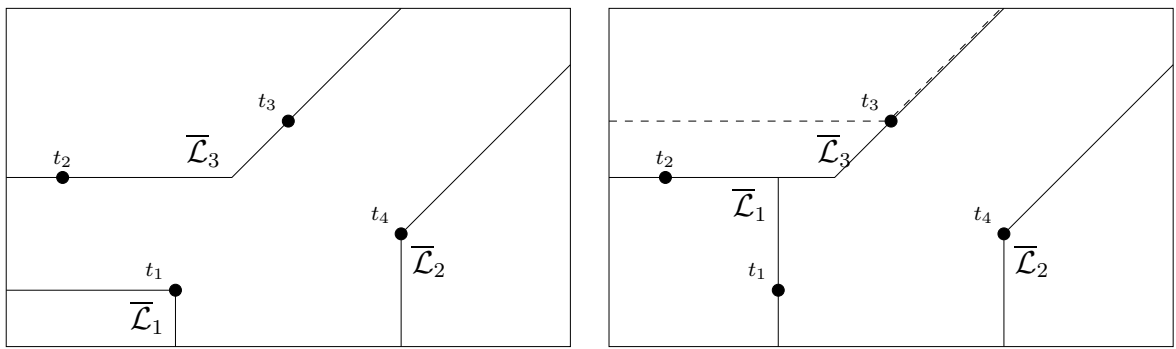
\begin{figure}[h!]
    \centering
    \begin{subfigure}[t]{0.45\textwidth}
        \centering
        \begin{tikzpicture}[scale=0.75]
\draw [line width=0.1mm](7,0) -- (7,2) -- (10,5);
\draw [line width=0.1mm](0,3) -- (4,3) -- (7,6);
\draw [line width=0.1mm](3,0) -- (3,1) -- (0,1);

\filldraw (7,2) circle (3pt);
\node[above left] at (7,2) {\tiny $t_4$};
\filldraw (5,4) circle (3pt);
\node[above left] at (5,4) {\tiny $t_3$};
\filldraw (3,1) circle (3pt);
\node[above left] at (3,1) {\tiny $t_1$};
\filldraw (1,3) circle (3pt);
\node[above] at (1,3) {\tiny $t_2$};

\node[above left] at (4,3) {$\mathcal{\overline{L}}_3$};
\node[below left] at (3,1) {$\mathcal{\overline{L}}_1$};
\node[below right] at (7,2) {$\mathcal{\overline{L}}_2$};
\draw[line width=0.1mm] (0,0) rectangle (10,6);
\end{tikzpicture}
        \caption{Realizable with outcome function $g$ where $g(t_1) = 1$, $g(t_2)=g(t_3) = 3$, $g(t_4) = 2$.}
    \end{subfigure}%
\hspace{1em}
    \begin{subfigure}[t]{0.45\textwidth}
        \centering
        \begin{tikzpicture}[scale=0.75]
\draw [line width=0.1mm](7,0) -- (7,2) -- (10,5);
\draw [line width=0.1mm](0,3) -- (4,3) -- (7,6);
\draw [line width=0.1mm, dashed](0,4) -- (4.95,4) -- (6.95,6);
\draw [line width=0.1mm](3,0) -- (3,3);

\filldraw (7,2) circle (3pt);
\node[above left] at (7,2) {\tiny $t_4$};
\filldraw (5,4) circle (3pt);
\node[above left] at (5,4) {\tiny $t_3$};
\filldraw (3,1) circle (3pt);
\node[above left] at (3,1) {\tiny $t_1$};
\filldraw (1,3) circle (3pt);
\node[above] at (1,3) {\tiny $t_2$};

\node[above] at (4,3) {$\mathcal{\overline{L}}_3$};
\node[below left] at (3,3) {$\mathcal{\overline{L}}_1$};
\node[below right] at (7,2) {$\mathcal{\overline{L}}_2$};
\draw[line width=0.1mm] (0,0) rectangle (10,6);
\end{tikzpicture}
        \caption{Not realizable: there is no $(1,3)$- nor $(3,1)$-witness, so $L$ does not separate $T$.}
    \end{subfigure}%
\caption{Realizable and non-realizable matrices, defined by the apices of $\overline{\mathcal{L}}_i$ for $i \in [3]$.
Figure accompanies Example \ref{ex:non.realizable}.}\label{fig:non.realizable}
\end{figure}

\begin{example}[Realizable and non-realizable matrices]\label{ex:non.realizable}
In Figure \ref{fig:non.realizable}, the set of four black points is the type space $T = \{t_1,t_2,t_3,t_4\}$. Each panel defines a matrix $L$ whose rows are the apices of the sectors $\overline{\mathcal{L}}_1, \overline{\mathcal{L}}_2$ and $\overline{\mathcal{L}}_3$ depicted using black lines. The matrix in Panel (\textsc{a}) is realizable: it equals $L^g$ for the mechanism $g$ defined by $g(t_2) = g(t_3) = 3$, $g(t_1) = 1$ and $g(t_4) = 2$.  The matrix $L$ in Panel (\textsc{b}) is not realizable. Suppose for a contradiction that it were realized by some outcome function~$g$. By Theorem~\ref{thm:realizable}, we would require $g(t_4) = 2$, $g(t_1) = 1$, $g(t_3) = 3$, and $g(t_2)$ taking value either 1 or 3. If $g(t_2) = 3$, then $L^g$ must equal the matrix shown in Figure \ref{fig:non.realizable}~(\textsc{a}) which is different from $L$. So $g(t_2) = 1$, but for this $g$, $\maxH(-L_3^g)$ equals the sector depicted with the dashed lines, and hence $L^g \neq L$. Thus $L$ is not realizable. Realizability fails since $L$ does not separate $T$ at $\{1,3\}$. Here, one has
$\mathcal{I}_{13} = \maxL_{13} \cap \maxL_{31} \cap T = \{t_2\}, $
but there is no $(1,3)$- nor a $(3,1)$-witness.
\end{example}

%% file: TGMDfigure_wm.tex
	\begin{figure}[h]
    \centering
    \begin{subfigure}[t]{0.33\textwidth}
        \centering
        \begin{tikzpicture}[scale=0.5]

\filldraw[black!10] (0,5.05)  to [out=10,in=215] (9.94,10)--(0,10);
\filldraw[black!10] (0,4.95)  to [out=0,in=90] (4.95,0)--(0,0);
\filldraw[black!10] (5.05,0)  to [out=90,in=235] (10,9.94)--(10,0);

\draw [line width=0.25mm, dotted](0,6) -- (6,6) -- (10,10);
\draw [line width=0.25mm, dotted](4,4) -- (6,6);
\draw [line width=0.25mm, dotted](6,6);
\draw [line width=0.25mm, dotted](4,0) -- (4,6);
\filldraw (4,6) circle (2.5pt);
\filldraw (6,6) circle (2.5pt);
\filldraw (4,4) circle (2.5pt);
\node[above left] at (6,6) {$\mathcal{\overline{L}}_3$};
\node[below right] at (4,4) {$\mathcal{\overline{L}}_2$};
\node[below left] at (4,6) {$\mathcal{\overline{L}}_1$};
\draw[line width=0.1mm] (0,0) rectangle (10,10);

\end{tikzpicture}
\caption{A mechanism that is weakly monotone but not IC.}
    \end{subfigure}%
    ~ 
    \begin{subfigure}[t]{0.33\textwidth}
        \centering
        \begin{tikzpicture}[scale=0.5]

\filldraw[black!10] (0,5.05)  to [out=10,in=215] (9.94,10)--(0,10);
\filldraw[black!10] (0,4.95)  to [out=0,in=90] (4.95,0)--(0,0);
\filldraw[black!10] (5.05,0)  to [out=90,in=235] (10,9.94)--(10,0);

\draw [line width=0.1mm](10,6) -- (4,6) -- (4,10);
\draw [line width=0.1mm](0,0) -- (4,4);
\draw [line width=0.1mm](4,4) -- (4,6);
\draw [line width=0.1mm,](4,4) -- (6,6);
\filldraw (4,6) circle (2.5pt);
\filldraw (6,6) circle (2.5pt);
\filldraw (4,4) circle (2.5pt);
\node[below right] at (6,6) {$\mathcal{\underline{L}}_3$};
\node[above left] at (4,4) {$\mathcal{\underline{L}}_2$};
\node[above right] at (4,6) {$\mathcal{\underline{L}}_1$};
\draw[line width=0.1mm] (0,0) rectangle (10,10);
\end{tikzpicture}
        \caption{The min-plus sectors verifying non-IC.}
    \end{subfigure}~
    \begin{subfigure}[t]{0.33\textwidth}
        \centering
        \begin{tikzpicture}[scale=0.5]
\filldraw[black!10] (0,5.05)  to [out=10,in=215] (9.94,10)--(0,10);
\filldraw[black!10] (0,4.95)  to [out=0,in=90] (4.95,0)--(0,0);
\filldraw[black!10] (5.05,0)  to [out=90,in=235] (10,9.94)--(10,0);
\draw [line width=0.25mm, dotted](10,10) -- (5,5) -- (0,5);
\draw [line width=0.28mm, dotted](5,5) -- (5,0);
\filldraw (5,5) circle (2.5pt);
\node[above] at (5,5) {$\mathcal{\underline{L}}_3$};
\node[below right] at (5,5) {$\mathcal{\underline{L}}_2$};
\node[below left] at (5,5) {$\mathcal{\overline{L}}_1$};
\draw[line width=0.1mm] (0,0) rectangle (10,10);
\end{tikzpicture}
        \caption{On the same type space, another mechanism that is RE.}
    \end{subfigure}
\caption{Demonstration of weak monotonicity, lack of incentive compatibility, on a revenue equivalent domain. Figure accompanies Example \ref{ex:non.convex}.}\label{fig:revenue-equivalent.non-convex}
\end{figure}
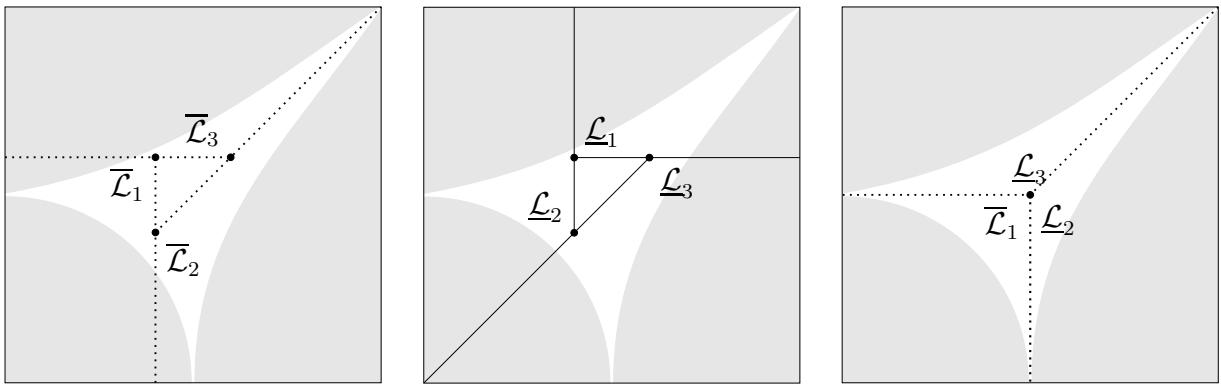

	\begin{example}[Weak monotonicity, IC and an RE type space] \label{ex:non.convex}
In Figure \ref{fig:revenue-equivalent.non-convex}, the set $T$ consists of all points in the gray shaded region. On this type space not all weakly monotone outcome functions on $T$ are IC. On convex type spaces this cannot happen \cf \cite{saks2005weak}. Panel~(\textsc{a}) defines a weakly monotone outcome function via $L^g$.
Weak monotonicity follows from Proposition~\ref{prop:weak.mon}, since the open sectors $\maxL_1^\circ, \maxL_2^\circ, \maxL_3^\circ$ are pairwise disjoint. 
Panel (\textsc{b}) verifies that $\Eig_0(L^g) = \emptyset$ via Corollary \ref{cor:L.payment}, so $g$ is not IC. Using the dimension graph one confirms that $T$ is RE using Corollary \ref{cor:re.type}. Note that if this type space were unbounded, then its closure would neither be path-connected, nor `boundedly grid-wise connected', which are other known and easily verified sufficient conditions for a type space to be RE, \cf \cite[Theorems 1 and 4]{chung2007non-differentiable}. See Section \ref{sec:general.types} on how to extend our theorems to such type spaces. 
\end{example}

%% file: TGMDfigure_multiplayer.tex
\begin{example}[Multiplayer Case]\label{ex:multiplayer} Let $n=2$, $T_1=T_2$ with $r_i=3$. Figure \ref{fig:multiplayer} depicts the arrangements $\minH(T_i)$.  For this two player example it is notationally convenient to use the following covector convention. We identify a covector matrix $\cov_{T_i}(q)$ of dimension $3\times 3$ with a vector  $(S_1, S_2, S_3)$ of subsets of $[3]$, with $k\in S_l$, if and only if $q \in \minH_k(-t_l)$, if and only if $\cov_{T_i}(q)_{l,t_k} = 1$. This identification extends to the $3\times 3 \times 3$ tensor corresponding to an outcome function giving a $3\times 3$ matrix whose entries are singleton subsets of $[3]$. Let $g:T_1\times T_2 \to [3]$ be an outcome-function. Upon identifying, Proposition \ref{prop:multiplayer} then states that each column and row of this $3\times 3$ matrix $g$ is a covector of a full-dimensional cell of the arrangement $\minH(-T_1)$ and $\minH(-T_2)$, respectively. 

Some examples of dominant strategy incentive compatible outcome functions are the following matrices, whose rows are labeled by the types of player 1,  $t_{1,1}, t_{1,2}, t_{1,3}$ and the columns by the types of player 2, $t_{2,1}, t_{2,2}, t_{2,3}$ $$ f= \begin{pmatrix} 3&3&1\\3&1&2\\1&2&2\end{pmatrix}, \quad g= \begin{pmatrix} 3&3&1\\3&1&1\\1&2&2\end{pmatrix},\quad  h= \begin{pmatrix} 3&3&3\\3&1&2\\3&1&2\end{pmatrix}.$$
\end{example}
Each column is a covector of the arrangement on $T_1$, each row is a covector of the arrangement on $T_2$, thus, $f, g, h$ are DIC.
\begin{figure}[h!]
    \centering
    \begin{subfigure}[t]{0.45\textwidth}
        \centering
        \begin{tikzpicture}[scale=0.75]
		
		
		\node at (3,5) {\tiny $(1,2,2)$};
		\node at (5.5,3) {\tiny $(3,1,2)$};
		\node at (8.5,1.5) {\tiny $(3,3,1)$};
		\node at (2.25,2.25) {\tiny $(3,2,2)$};
		
		\filldraw (2,4) circle (2.5pt);
		\draw [line width=0.1mm](0,2) -- (2,4) -- (2,6) -- (2,4) -- (10,4);
		\node[below right] at (2,4) {$t_{1,1}$};
		
		\filldraw (4,2) circle (2.5pt);
		\draw [line width=0.1mm](2,0) -- (4,2) -- (4,6) -- (4,2) -- (10,2);
		\node[above right] at (4,2) {$t_{1,2}$};
		
		\filldraw (7,1) circle (2.5pt);
		\draw [line width=0.1mm](6,0) -- (7,1) -- (7,6) -- (7,1) -- (10,1);
		\node[above left] at (7,1) {$t_{1,3}$};

		\draw[line width=0.1mm] (0,0) rectangle (10,6);
\end{tikzpicture}
        \caption{Arrangement of Player 1 on $T_1$.}
    \end{subfigure}%
    \hspace{0.5cm} 
    \begin{subfigure}[t]{0.45\textwidth}
        \centering
        \begin{tikzpicture}[scale=0.75]
		
		\node at (3,5) {\tiny $(1,2,2)$};
		\node at (5.5,5) {\tiny $(1,1,2)$};
		\node at (5.5,3) {\tiny $(3,1,2)$};
		\node at (8.5,3) {\tiny $(3,1,1)$};
		\node at (8.5,1.5) {\tiny $(3,3,1)$};
		\node at (8.5,0.5) {\tiny $(3,3,3)$};
		
		\filldraw (2,4) circle (2.5pt);
		\draw [line width=0.1mm](0,2) -- (2,4) -- (2,6) -- (2,4) -- (10,4);
		\node[below right] at (2,4) {$t_{2,1}$};
		
		\filldraw (4,2) circle (2.5pt);
		\draw [line width=0.1mm](2,0) -- (4,2) -- (4,6) -- (4,2) -- (10,2);
		\node[above right] at (4,2) {$t_{2,2}$};
		
		\filldraw (7,1) circle (2.5pt);
		\draw [line width=0.1mm](6,0) -- (7,1) -- (7,6) -- (7,1) -- (10,1);
		\node[above left] at (7,1) {$t_{2,3}$};

		\draw[line width=0.1mm] (0,0) rectangle (10,6);
\end{tikzpicture}
        \caption{Arrangement of Player 2 on $T_2$. }
    \end{subfigure}
\caption{Here $T_1 = \{t_{1,1}, t_{1,2}, t_{1,3}\}$ and $T_2 = \{t_{2,1}, t_{2,2}, t_{2,3}\}$. The min-plus arrangements $\minH(-T_1)$ and $\minH(-T_2)$ are depicted in the left and right panels respectively.}\label{fig:multiplayer}
\end{figure}
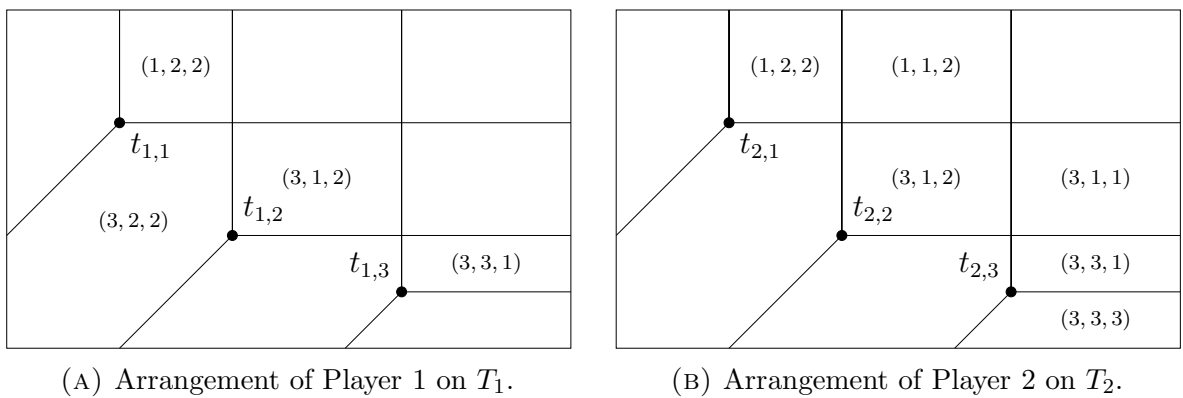

%% file: TGMDfigure_noncompact.tex
\begin{example}[Type spaces with compact closure]\label{ex:compact.closure}
 Consider the type space $$T=\{(0, 1/n, 1/n):n \in\N\}\cup \{(0,-1,0)\}\}. $$ As a subset of $\R^3$ this type set is not closed, since $(0,0,0)$ is a limit point which is not in $T$. In fact, $T\cup \{ (0,0,0)\}$ is the closure of $T$ and compact. Now suppose we were to define $\tconv(T)$ as we did in Definition \ref{def:tropical.hull}. Observe that this is no longer a polyhedral complex. Indeed, the polyhedron $\textsf{conv}((0,-1,0),(0,0,0))+\R\cdot(1, \ldots, 1)$, where $\textsf{conv}$ denotes the ordinary convex hull, is a face of $\sigma = \left(\bigcap_{n\in\N}\minH_1(0,1/n,1/n)\right)\cap\minH_2(0,-1,0)\in \minH(-T)$ but not contained in $\minH(-t)$. This face is only a cell upon considering the intersection of all sectors in the arrangement on the closure $T\cup \{(0,0,0)\}$. 

 To deal with type spaces whose closure is compact, two ways present themselves, (i) applying Definition \ref{def:tropical.hull}to the topological closure of $T$, (ii) defining $\tconv(T)$ as the smallest polyhedral complex containing all bounded cells of $\minH(-T)$. In fact, one can show that the tropical polytope in the sense of Definition \ref{def:tropical.hull} of the closure of a bounded type space agrees with the smallest polyhedral complex on $T$ containing the intersection of all sectors.  

 While either approach is a valid way forward, we emphasize one subtlety. The definition of the  dimension-graph needs to be adapted if proceeding as in (ii). To see this, note that $(0,0,0)$ would have to be a cell in either case, but this point is present only in the closure of $T$, so the dimension graph of $(0,0,0)$ would be connected if the tropical polytope were defined as in (i), but not if it is defined as in (ii). This problem is resolved by defining the dimension graph as the \emph{undirected} graph on $m$ nodes with edge $(i,j)$ if and only if $$\textsf{dist}(\minH_{ij}(-p), \minH_i(-p)\cap T)=0.$$ It is easy to see that this agrees with (\ref{eqn:dimension.graph}) when $T$ is compact, but is well defined for any type space. Adapting our theorems to the strongly connected components of the dimension graph is then straight forward. 
\end{example}

\begin{example}[Type spaces without compact closure]\label{ex:noncompact.closure} Suppose the type space is given by $$T=\{(0,n,m):n,m\in \N \}.$$ Clearly, any mechanism that is not onto requires some prices to be infinite. Thus we must consider the closure of $T$ in $(\R\cup \{+\infty\})^m\backslash \{(+\infty, \ldots,+\infty)\}.$ The notion of tropical polytopes and covectors of \cite{develin2004tropical} generalizes to this case, see \cite{fink2013stiefel, joswig2016weighted}. In particular \cite[Section 3.5]{joswig2016weighted} defines a covector in this case, which can be used to extend the present theory. The obstacle to overcome is that the notion of ``bounded cells" to apply Definition \ref{def:tropical.hull} is ill defined for any cell cells containing points with infinite coordinates. This problem can be resolved by viewing the ``boundary", i.e. the set of points with infinite coordinates as  lower-dimensional affine copies of $\R^d$ for $d<m$. Since this extension requires notation that is somewhat involved we prefer to restrict attention to the compact case in this paper. We wish to point out one subtlety of tropical arithmetic however. While we could use the $\min$ and $\max$ convention interchangeably in $\TA^{m-1}$, this is no longer the case when admitting infinite values, since the generalization of tropical affine space, called the \emph{tropical projective torus} depends on choice of tropical addition (by adding $+\infty$ instead of $-\infty$ which are the additive identities for $\underline \oplus$ and $\overline\oplus$, respectively). 
\end{example}

%% file: TGMDfigure_limit.tex
\begin{example}[Limiting approximation] \label{ex:limit}
	The set $T \subset \TA^2$ in Figure \ref{fig:re} Panel (\textsc b) consists of the three parallel lines. This set can be approximated as the limit of a finite sequence $(T^k)$, we  schematically depicted $T^{21}$ by the black points in Panel \textsc{(a)}, together with tropical polytope $\tconv(T^{21})$ shaded gray. Panel \textsc{(b)} depicts the limit. It can be shown that the tropical polytopes converge as polyhedral complexes in the Hausdorff distance based on the Euclidean distance. 
	\end{example}

%% file: TGMDfigure_generic.tex
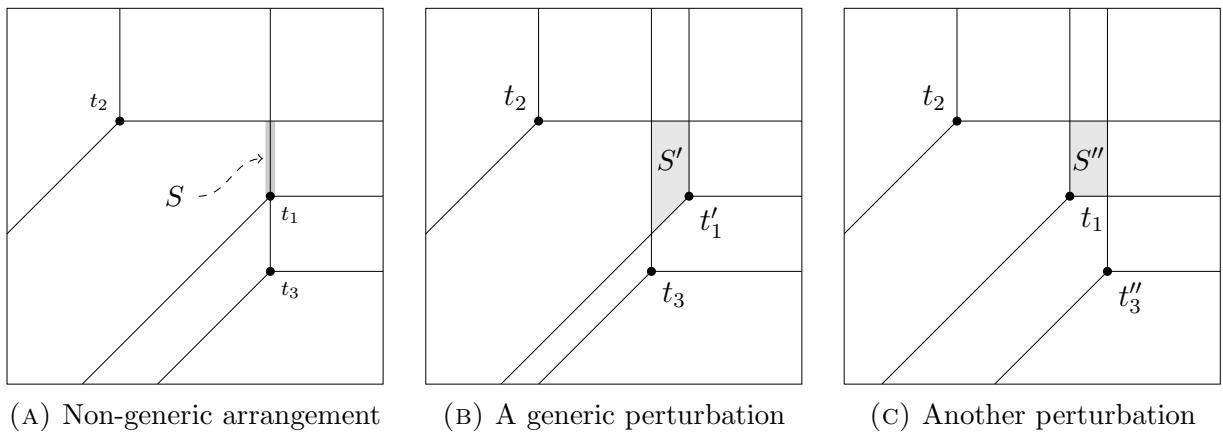
\begin{figure}
\centering
\begin{subfigure}{0.33\textwidth}
\centering
\begin{tikzpicture}[scale=0.5]
	\draw [line width=1.3mm, black!20] (7,5) -- (7,7);
	\draw [line width=0.1mm](3,10) -- (3,7) -- (10,7); 
	\draw [line width=0.1mm] (3,7) -- (0,4);
	\draw [line width=0.1mm](7,10) -- (7,3)-- (4,0);
	\draw [line width=0.1mm](2,0) -- (7,5) -- (10,5);
	\draw [line width=0.1mm](7,3) -- (10,3);
	
	\filldraw (3,7) circle (3pt);
	\filldraw (7,5) circle (3pt);
	\filldraw (7,3) circle (3pt);
	\node[below right] at (7,5) {\tiny $t_1$};
	\node[below right] at (7,3) {\tiny $t_3$};
	\node[above left] at (3,7) {\tiny $t_2$};
	\draw[line width=0.1mm] (0,0) rectangle (10,10);
	
	\draw
    [->,thin, dashed](5.1,5) to [out=0, in=180] (6.8,6);
    \node[left] at (5,5) {\small $S$};
	
\end{tikzpicture}
\caption{Non-generic arrangement}
\end{subfigure}
    ~ 
\begin{subfigure}{0.33\textwidth}
\centering
\begin{tikzpicture}[scale=0.5]
	\filldraw [black!10](6,4) -- (7,5) -- (7,7) --(6,7); 
	\draw [line width=0.1mm](3,10) -- (3,7) -- (10,7); 
	\draw [line width=0.1mm] (3,7) -- (0,4);
	\draw [line width=0.1mm](6,10) -- (6,3)-- (3,0);
	\draw [line width=0.1mm](2,0) -- (7,5) -- (10,5);
	\draw [line width=0.1mm](6,3) -- (10,3);
	\draw [line width=0.1mm](7,5) -- (7,10);

	\filldraw (3,7) circle (3pt);
	\filldraw (7,5) circle (3pt);
	\filldraw (6,3) circle (3pt);
	\node[below right] at (7,5) {$t'_1$};
	\node[below right] at (6,3) {$t_3$};
	\node[above left] at (3,7) {$t_2$};
	\draw[line width=0.1mm] (0,0) rectangle (10,10);
	\node at (6.5,6) {\small $S'$};
\end{tikzpicture}
\caption{A generic perturbation}
\end{subfigure}
~
\begin{subfigure}{0.33\textwidth}
\centering
\begin{tikzpicture}[scale=0.5]
	\filldraw [black!10](6,5) -- (7,5) -- (7,7) --(6,7); 
	\draw [line width=0.1mm](3,10) -- (3,7) -- (10,7); 
	\draw [line width=0.1mm] (3,7) -- (0,4);
	\draw [line width=0.1mm](7,10) -- (7,3)-- (4,0);
	\draw [line width=0.1mm](1,0) -- (6,5) -- (10,5);
	\draw [line width=0.1mm](6,5) -- (6,10);
	\draw [line width=0.1mm](7,3) -- (10,3);
	
	\filldraw (3,7) circle (3pt);
	\filldraw (6,5) circle (3pt);
	\filldraw (7,3) circle (3pt);
	\node[below right] at (6,5) {$t_1$};
	\node[below right] at (7,3) {$t''_3$};
	\node[above left] at (3,7) {$t_2$};
	\draw[line width=0.1mm] (0,0) rectangle (10,10);
	\node at (6.5,6) {\small $S''$};
\end{tikzpicture}
\caption{Another perturbation}
\end{subfigure}
	
\caption{A non-generic arrangement on three types and possible generic perturbations thereof. The cells shaded gray support onto mechanisms. Figure accompanies Example \ref{ex:generic.perturbation}.}\label{fig:generic.perturbation}
\end{figure}

\begin{example}[Generic perturbation of finite point-configurations] \label{ex:generic.perturbation} Figure~\ref{fig:generic.perturbation} depicts arrangements on three points. The points in Panel \textsc{(a)} are not generic: the projections of $t_1$ and $t_3$ onto the first coordinate is a point. Panels \textsc{(b)} and \textsc{(c)} show possible generic perturbations of these points. The covectors of the cells shaded gray are \begin{equation*}
\begin{matrix}
	\underline{\mathsf{coVec}}_T(S)=\begin{pmatrix}
1&0&1\\ 1&0&1\\0 &1&0	
\end{pmatrix}, & 
\mathsf{\underline{coVec}}_L(S')=\begin{pmatrix}
0&0&1\\ 1&0&0\\0&1&0	
\end{pmatrix}, &
\mathsf{\underline{coVec}}_L(S'')=\begin{pmatrix}
1&0&0\\ 0&0&1\\0&1&0	
\end{pmatrix}
\end{matrix}	
\end{equation*}
The cells $S'$ and $S''$ are full-dimensional and thus support surjective outcome functions by Corollary \ref{cor:cell.mechanism.bijection}. Indeed, their covectors are the outcome functions $g$ and $h$ respectively, with $g(t'_1)=2, g(t_2)=3, g(t_3)=1$ and $h(t_1)=1, h(t_2)=3, h(t_3'')=2$. The cell $S$ in Panel \textsc{(a)} can be seen as a limit of $S'$ as $t_1'$ approaches $t_1$, or as the limit of $S''$ as $t_3''$ approaches $t_3$. It supports the surjective outcome functions $g$ and $h$, that is $P(\nu(g)) = P(\nu(h)) = P(\cov_T(S))$. \end{example}